\documentclass[11pt]{article}
\usepackage{geometry,amsmath,amssymb,enumerate,color,theorem}
\newcommand{\bignone}{}
\newcommand{\email}[1]{{\textit{Email:} \texttt{#1}}}
\newcommand{\tmem}[1]{{\em #1\/}}
\newcommand{\tmop}[1]{\ensuremath{\operatorname{#1}}}
\newcommand{\tmtextit}[1]{{\itshape{#1}}}
\newcommand{\tmtextsc}[1]{{\scshape{#1}}}
\newenvironment{enumeratealpha}{\begin{enumerate}[a{\textup{)}}] }{\end{enumerate}}
\newenvironment{enumerateroman}{\begin{enumerate}[i.] }{\end{enumerate}}
\newenvironment{proof}{\noindent\textbf{Proof\ }}{\hspace*{\fill}$\Box$\medskip}
\definecolor{grey}{rgb}{0.75,0.75,0.75}
\definecolor{orange}{rgb}{1.0,0.5,0.5}
\definecolor{brown}{rgb}{0.5,0.25,0.0}
\definecolor{pink}{rgb}{1.0,0.5,0.5}
\newtheorem{corollary}{Corollary}
{\theorembodyfont{\rmfamily}\newtheorem{example}{Example}}
\newtheorem{lemma}{Lemma}
\newtheorem{proposition}{Proposition}
\newtheorem{theorem}{Theorem}

\newcommand{\maxent}{\tmtextsc{MaxEnt}}
\numberwithin{equation}{section}
\numberwithin{table}{section}
\numberwithin{figure}{section}
\numberwithin{lemma}{section}
\numberwithin{proposition}{section}
\numberwithin{corollary}{section}

\begin{document}

\title{Analytical Forms for Most Likely Matrices Derived from Incomplete
Information}
\author{Kostas N. Oikonomou\thanks{\email{ko@research.att.com}}\\
AT\&T Labs Research\\
Middletown, NJ, 07748, U.S.A.}

\maketitle

\begin{abstract}
  Consider a rectangular matrix describing some type of communication or
  transportation between a set of origins and a set of destinations, or a
  classification of objects by two attributes. The problem is to infer the
  entries of the matrix from limited information in the form of constraints,
  generally the sums of the elements over various subsets of the matrix, such
  as rows, columns, etc, or from bounds on these sums, down to individual
  elements. Such problems are routinely addressed by applying the maximum
  entropy method to compute the matrix numerically, but in this paper we
  derive analytical, closed-form solutions. For the most complicated cases we
  consider the solution depends on the root of a non-linear equation, for
  which we provide an analytical approximation in the form of a power series.
  Some of our solutions extend to 3-dimensional matrices.
  
  Besides being valid for matrices of arbitrary size, the analytical solutions
  exhibit many of the appealing properties of maximum entropy, such as precise
  use of the available data, intuitive behavior with respect to changes in the
  constraints, and logical consistency.
\end{abstract}

\section{Introduction}

Consider a set of $n$ origins communicating with a set of $m$ destinations.
For our purposes it suffices that each origin is connected to each
destination; the exact nature of the connection is not important. The
communication may be in the form of transportation, e.g. the origins and
destinations may be cities or other geographic locations, and people travel
from one to another by some means, or commodities are transported from one to
another in some fashion. Or the origins and destinations may be nodes
connected by a communications network, with various sorts of traffic flowing
from each source to each destination. In either of these cases, the
transportation or communication can be represented by an rectangular $n \times
m$ {\tmem{trip}} or {\tmem{traffic}} matrix whose $i, j$th entry gives the
number of trips, volume of traffic, units of a commodity, etc. from the $i$th
origin to the $j$th destination. (The distinction between origins and
destinations is not mandatory; one could take $n = m$ and think just of a set
of $n$ locations.) In a different setting, we have a set of objects with two
attributes, say height and weight, color and shape, success/failure of a test
and test condition, and the objects are placed in a table according to the
$n$-valued first attribute and the $m$-valued second attribute. In this
setting the $n \times m$ matrix is known as a (2-dimensional)
\tmtextit{contingency table} whose $(i, j)$th entry is the number of objects
whose 1st attribute has the $i$th value and 2nd attribute the $j$th value.

Whichever of these two settings obtains, we are interested in the situation
where we have {\tmem{limited}} or incomplete information about the matrix: we
do not know the individual elements, but know less detailed characteristics
such as the totals of the rows and/or columns, or of some of them, the total
sum of the matrix, or we have bounds on some of these quantities, or in
addition we know the values of some individual elements or have bounds on
them. The problem then is how to \tmtextit{infer} all the matrix elements from
this information, and, in this paper, we are interested in solving the problem
analytically. It is well known how to find numerical solutions to these
inference problems by numerical entropy maximization.

\paragraph{Most likely matrices and maximum entropy}

We approach the problem by regarding the matrix as constructed from a known
number of elements (trips, traffic units, etc), which we will think of as
balls, to be placed into an $n \times m$ array of boxes. We will refer to the
number of ways (assignments of balls to boxes) in which a given matrix $X$ can
be built as its \tmtextit{number of realizations}, $\#(X)$. If the information
$I$ is known about $X$, we may also regard it as {\tmem{constraints}} that $X$
has to satisfy, and we write $\#(X|I)$ for the number of realizations of $X$
that accord with $I$ or satisfy the constraints $I$. For example, if what we
know about the $2 \times 2$ matrix $X$, $x_{i j} \in \mathbb{N}$, is that its
row sums are 7 and 3, some possibilities are
\[ X_1 = \left(\begin{array}{cc}
     3 & 4\\
     2 & 1
   \end{array}\right), \hspace{1em} X_2 = \left(\begin{array}{cc}
     4 & 3\\
     2 & 1
   \end{array}\right), \hspace{1em} X_3 = \left(\begin{array}{cc}
     4 & 3\\
     1 & 2
   \end{array}\right), \hspace{1em} X_4 = \left(\begin{array}{cc}
     2 & 5\\
     2 & 1
   \end{array}\right), \hspace{1em} X_5 = \left(\begin{array}{cc}
     1 & 6\\
     3 & 0
   \end{array}\right) . \]
In fact there are 8 possible 1st rows and 4 possible 2nd rows, so 32 matrices
satisfy these constraints. Further, for the above examples, $\#(X_1 |I)
=\#(X_2 |I) =\#(X_3 |I) = 10! / (1! 2! 3! 4!) = 12600, \; \#(X_4 |I) = 10! /
(5! 2! 2! 1!) = 7560, \; \#(X_5 |I) = 10! / (1! 3! 6!) = 840$ {\footnote{We
assume that the balls are distinguishable. The boxes are distinguishable,
being particular elements of a matrix.}}. We will refer to the matrix
$\hat{X}$ for which $\#(X|I)$, given by a multinomial coefficient, is maximum,
as the \tmtextit{most likely} matrix given the information/constraints $I$.
``Most likely'' may have probability connotations for some, but we use it only
as a shorthand for ``matrix that can be realized in the greatest number of
ways'', which has nothing to do with probability, it is merely counting.

If the constraints are convex (in this paper they will be linear), and they
specify the sum of all the elements, the discrete problem of maximizing
$\#(X|I)$ can be turned into a continuous concave maximization problem via the
Stirling approximation to the factorial: the log of the multinomial
coefficient is approximated by the \tmtextit{entropy} of the $x_{i j}$. This
continuous approximation is well-known, and in fact dates back to Boltzmann's
(1847-1906) combinatorial formulation of statistical mechanics where molecules
are assigned to boxes; see e.g. {\cite{Som}}. Thus our discrete most likely
matrix problem connects to the extensive body of work on maximum entropy
({\maxent}): see the works {\cite{JaynesCP}}, {\cite{JL}} of E. T. Jaynes, the
books {\cite{Tr69}}, {\cite{KK}}, and the series of {\maxent} conference
proceedings {\cite{maxent-kluwer}} and {\cite{maxent-aip}}{\footnote{The
latter with the unfortunate adoption of Microsoft Word for the typesetting of
mathematical papers.}}, to name a few. So, as long as the total sum is known,
the discrete most likely matrix problem and its continuous {\maxent} analogue
are equivalent to within the Stirling approximation, and we will sometimes
refer to one, sometimes to the other.

The combinatorial rationale that we consider here is appealing because of its
simplicity: it is just counting. In addition, {\maxent} has intuitive appeal
as maximizing uncertainty while conforming to precisely the available
information. More importantly, it has a powerful {\tmem{axiomatic}} basis as
well: see {\cite{Skilling1989}}, and {\cite{Caticha2006}} for recent
developments.

\paragraph{Summary and background}

In this paper we derive analytical, closed-form solutions to a set of maximum
entropy problems having to do with $n \times m$ matrices subject to linear
constraints. The constraints have the form of equalities or inequalities
(upper bounds) on sums over various subsets of the matrix, e.g. rows, columns,
the whole matrix, the diagonal, individual elements, etc. In {\S}\ref{sec:rcs}
to {\S}\ref{sec:boundsrc} we consider known row, column, and total sums, as
well as upper bounds on them. We observe that when the total sum is not known,
the most likely matrix is not the {\maxent} matrix, but it has a simple
relationship to a certain {\maxent} matrix. In {\S}\ref{sec:bindiv} we
consider upper bounds on row sums and on individual elements. Finally, in
{\S}\ref{sec:symm}, we investigate the effect of having symmetric information
in combination with bounds on sums and specified individual elements,
including an extension to 3-dimensional matrices. Table \ref{tab:summary} in
{\S}\ref{sec:concl} summarizes the types of constraints that we consider. In
the most complicated cases the solutions depend on the root of a single
non-linear equation, but even in those cases we find an analytical power
series approximation to the root, hence to the matrix elements themselves. The
analytical forms allow us to treat matrices of arbitrary size, reveal the
exact structure of the most likely/{\maxent} matrix, and allow us to see
explicitly the robustness of the solution to changes in the constraints, and
its behavior with respect to uncertainty in the data. These features
demonstrate the {\tmem{logical precision}} of the {\maxent} method and are
inaccessible via numerical solutions.

An extensive and in-depth study of {\maxent} matrices in transportation
analysis is {\cite{grav}}; an introduction can be found in {\cite{KK}}, and a
recent reference is {\cite{Bilich2008}}. Various aspects of matrices
characterizing traffic in IP networks, including numerical estimation from
incomplete data, are studied in {\cite{Alderson2006}}, {\cite{ZRLD}}, and the
references therein{\footnote{We say ``estimation'' because the methods used do
not have the same logical standing as {\maxent}. Also, the problem of
estimating traffic in a real IP network is significantly more complex than the
problems considered in this paper.}}. A semi-analytical derivation of most
likely traffic matrices subject to a total cost constraint is in {\cite{KO}}.
With respect to contingency tables, {\cite{KK}} provides an introduction while
{\cite{GoodME}} derives fundamental results on the ``vanishing of
interactions'' in {\maxent} multi-dimensional tables. For a small sample of
other applications of {\maxent} see {\cite{Cmiel2002}} and
{\cite{Sengupta1991}} (economics and econometrics), {\cite{Kouvatsos1992}}
and {\cite{Ku-Mahamud1993}} (queueing problems), and {\cite{Trivedi2002}}
(systems theory).

\section{Specified row sums and some column sums}

\label{sec:rcs}We begin by considering a small extension of a problem whose
solution is already known in the literature in order to introduce the concepts
and general methodology used in the rest of the paper. Phrasing the discussion
in terms of an $n \times m$ matrix $X$ describing the traffic from $n$ origins
to $m$ destinations, suppose we have the following information (or
constraints) $I$ about it:
\begin{enumerate}
  \item The total traffic from each origin: $\forall i, \sum_j \bignone x_{i
  j} = u_i$.
  
  \item The total traffic to each of the first $\ell \leqslant m$
  destinations: $\forall j \leqslant \ell, \sum_i x_{i j} = v_j \bignone$
\end{enumerate}
We assume that the information is consistent, i.e.$\sum_i u_i \geqslant \sum_j
v_j$. This information also specifies the total traffic $s$ in the network: $s
= \sum_i u_i$.

To find the most likely traffic matrix $\hat{X}$ that follows from the
information $I$, given that the sum of all the entries is $s$, we construct
$X$ by distributing the $s$ units of traffic into $n m$ boxes so that $x_{i
j}$ of them go in box $(i, j)$. The number of ways in which this can be done
is
\begin{equation}
  \label{eq:Nd} \left(\begin{array}{c}
    s\\
    x_{11}, \ldots, x_{1 m}, x_{21}, \ldots, x_{2 m}, \ldots, x_{n 1}, \ldots,
    x_{n m}
  \end{array}\right) = \frac{s !}{\prod_{i, j} x_{i j} !} = \; \#(X \mid s),
\end{equation}
where the notation indicates that $s$ is known. To render the maximization of
$\#(X \mid s)$ tractable, and, at the same time, achieve a relatively simple
solution, we treat it as a {\tmem{continuous}} problem and maximize the log of
$\#(X \mid s)$ using the Stirling approximation
\begin{equation}
  \label{eq:stirling} \ln x! \; = \; x \ln x - x + \frac{1}{2} \ln x + \ln
  \sqrt{2 \pi} + \frac{\vartheta}{12 x}, \hspace{1em} \vartheta \in (0, 1),
\end{equation}
which is defined for all $x > 0$ by $x! = \Gamma (x + 1)$. Using the first two
terms of (\ref{eq:stirling}) and noting that $\sum_{i, j} x_{i j} = s$ is
given, the problem becomes
\begin{equation}
  \label{eq:maxent} \text{maximize} \hspace{2em} - \sum_{i, j} x_{i j} \ln
  x_{i j},
\end{equation}
subject to
\[ \sum_j x_{i j} = u_i \hspace{1em} \text{for} \; i = 1, \ldots, n,
   \hspace{2em} \sum_i x_{i j} = v_j \hspace{1em} \text{for} \; j = 1, \ldots,
   \ell . \]
The expression to be maximized is the {\tmem{entropy}} of the set of demands
$x_{i j}$. (Usually, e.g. in information theory, entropy is defined for a
vector whose entries sum to 1. What we use here is more properly referred to
as {\tmem{combinatorial}}, as opposed to {\tmem{information}},
entropy{\footnote{This usage goes back to Boltzmann's combinatorial
formulation of statistical mechanics, see {\cite{Som}}.}}.) Because the
entropy is a strictly concave function, the problem (\ref{eq:maxent}) has a
unique solution which can be found by forming the Lagrangean (details in
\S\ref{sec:concave})
\[ \Phi \; = \; - \sum_{i, j} x_{i j} \ln x_{i j} - \sum_i \lambda_i \Bigl(
   \sum_j x_{i j} - u_i \Bigr) - \sum_j \mu_j \Bigl( \sum_i x_{i j} - v_j
   \Bigr) . \]
It follows that $x_{i j} = e^{- \lambda_i - \mu_j - 1}$ if $j \leqslant \ell$
and $e^{- \lambda_i - 1}$ if $j > \ell$. Denoting $e^{- \lambda_i - 1}$ by
$\lambda'_i$ and $e^{- \mu_j}$ by $\mu_j'$, and then eliminating the primes to
simplify the notation,
\begin{equation}
  \label{eq:grdij} x_{i j} = \left\{ \begin{array}{ll}
    \lambda^{}_i \mu_j, & j \leqslant \ell\\
    \lambda_i, & j > \ell
  \end{array} \right., \hspace{2em} \lambda_i, \mu_j > 0.
\end{equation}
The origin and destination constraints imply that
\begin{equation}
  \begin{array}{cccc}
    \lambda_i (\mu_1 + \cdots + \mu_{\ell} + m - \ell) & = & u_i, & i = 1,
    \ldots, n\\
    (\lambda_1 + \cdots + \lambda_n) \mu_j & = & v_j, & j = 1, \ldots, \ell .
  \end{array} \label{eq:grtemp}
\end{equation}
Adding the first set of constraints together and doing the same with the
second set we get $(\lambda_1 + \cdots + \lambda_n) (\mu_1 + \cdots +
\mu_{\ell} + m - \ell) = u_1 + \cdots + u_n = s$ and $(\lambda_1 + \cdots +
\lambda_n) (\mu_1 + \cdots + \mu_{\ell}) = v_1 + \cdots + v_{\ell}$. If we now
let $\lambda$ be the sum of the $\lambda_i$ and $\mu$ that of the $\mu_j$, it
follows that if $\ell < m$
\[ \lambda = \frac{s - (v_1 + \cdots + v_{\ell})}{m - \ell}, \hspace{1em} \mu
   = \frac{(n - \ell) (v_1 + \cdots + v_{\ell})}{s - (v_1 + \cdots +
   v_{\ell})} . \]
So from (\ref{eq:grtemp}),
\begin{equation}
  \label{eq:grd} \lambda_i = \frac{\bigl( s - (v_1 + \cdots + v_{\ell}) \bigr)
  u_i}{(n - \ell) s}, \hspace{1em} \mu_j = \frac{(m - \ell) v_j}{s - (v_1 +
  \cdots + v_{\ell})} .
\end{equation}
Using this in (\ref{eq:grdij}), we finally get
\begin{equation}
  \label{eq:gm} \hat{x}_{i j} = \left\{ 
  \begin{array}{ll}
    \displaystyle
    \frac{u_i v_j}{s}, & j \leqslant \ell,\\
    \displaystyle
    \frac{s - (v_1 + \cdots + v_{\ell})}{m - \ell}  \frac{u_i}{s}, & j > \ell,
  \end{array} \right. \hspace{1em} i = 1, \ldots, n.
\end{equation}
Now if $\ell = m$, (\ref{eq:grtemp}) and the two equations following it become
$\lambda_i \mu = u_i$, $\mu_j \lambda = v_j$, $\lambda \mu = s$, from which it
follows that $\lambda_i \mu_j = u_i v_j / s$; thus (\ref{eq:gm}) is valid even
when $\ell = m$. Therefore the most likely matrix $\hat{X}$ consists of an $n
\times \ell$ left-hand part whose entries are given in the 1st line of
(\ref{eq:gm}), and a possibly empty right-hand part consisting of $m - \ell$
identical columns, each of which is described by the 2nd line of
(\ref{eq:gm}).

The solution $\hat{x}_{i j} = u_i v_j / s$ for all $i, j$ is known as the
{\tmem{gravity model}} for the traffic. This model has its origins in
transportation analysis, in connection with the numbers of trips taken between
$n$ origins and $m$ destinations, which are cities of known populations; $X$
is then referred to as a ``trip matrix''. See {\cite{KK}} for an introduction,
and {\cite{grav}} for an in-depth treatment. In the context of contingency
tables this model is known as the ``independence model'' under marginal
constraints. An important generalization to \tmtextit{multi-dimensional} $n_1
\times n_2 \times \cdots$ contingency tables is given in the classic paper
{\cite{GoodME}} of I. J. Good.

The form of $\hat{X}$ is {\tmem{conceptually robust}}. For example, take the
model with $n = m = \ell$ and suppose the destination constraints are removed.
Then all the $\mu_j'$ in (\ref{eq:grdij}) can be taken equal to 1, and the
solution is $\hat{x}_{i j} = u_i / n$, $\forall j$. Similarly, if the source
constraints are removed, $\hat{x}_{i j} = v_j / n$, $\forall i$. And if both
types of constraints are removed leaving just $\sum_{i, j} x_{i j} = s$, then
$\hat{x}_{i j} = s / n^2$. We see that {\maxent} yields independence and as
much symmetry/uniformity as possible, subject to the given information.

\section{Bounds on row sums}

\label{sec:bounds2}Suppose that the only information we have on the $n \times
m$ matrix $X$ is upper bounds on the row sums:
\begin{equation}
  \label{eq:bo} \forall i, \hspace{1em} \sum_j x_{i j} \leqslant u_i .
\end{equation}
We will first show that with $x_{i j} \in \mathbb{N}$, the most likely matrix
$\hat{X}$ has its row sums in fact \tmtextit{equal} to $u_1, \ldots, u_n$.
Indeed, let $X$ be a matrix satisfying the constraints (\ref{eq:bo}), and with
$\sum_{i, j} x_{i j} = s$. Suppose that row $i$ sums to strictly less than
$u_i$. This means that there is a $j$ such that if we increase $x_{i j}$ by 1,
the resulting matrix $X'$ also satisfies the constraints. By (\ref{eq:Nd}),
$X'$ is more likely than $X$:
\[ \frac{\#(X')}{\#(X)} \; = \; \frac{(s + 1) !}{s !} \; \frac{x_{i j}
   !}{(x_{i j} + 1) !} \; = \; \frac{s + 1}{x_{i j} + 1} \; > \; 1. \]
Proceeding in this way we can keep increasing the elements of the matrix while
also increasing the value of $\#(X)$, until all constraints are satisfied with
equality and the rows sum to exactly $u_1, \ldots, u_n$. This reduces the
problem to the one considered in \ref{sec:rcs}, where the total demand from
each origin is known (as well as the total demand in the whole network). So
the solution to (\ref{eq:bo}) is simply
\begin{equation}
  \label{eq:bo2} \forall i, \hspace{1em} \hat{x}_{i j} = \frac{u_i}{n} .
\end{equation}
This answer depends exactly on the given information and on nothing else. The
argument we gave above also shows that {\tmem{lower}} bounds on the row sums
are immaterial.

\begin{example}
  {\color{red} \label{ex:ub}}Suppose we have a $10 \times 10$ matrix, and the
  upper bounds on the row sums are $u_1, \ldots, u_{10} = 20, 20, 24, 30, 30,
  36, 36, 36, 36, 40$, measured in some units. We then find that the total
  number of matrices that accord with the information $I$ is
  \[ M (I) \; = \; 30045015^2 \cdot 131128140 \cdot 847660528^2 \cdot
     4076350421^4 \cdot 10272278170 \; \approx \; 2.41 \cdot 10^{89} . \]
  (The number of solutions in $\mathbb{N}$ of the equation $x + x_2 + \cdots
  + x_{10} = u_i$ is simply the number of {\tmem{compositions}} of $u_i$ into
  10 parts, equal to $\binom{u_i + 9}{9}$. And the inequality version can be
  handled by summing $\binom{b + 9}{9}$ over $0 \leqslant b \leqslant u_i$.)
  The most likely matrix $\hat{X}$ is one of these $2.4 \cdot 10^{89}$
  matrices. We can also find the number of matrices that satisfy (\ref{eq:bo})
  with equality. This turns out to be
  \[ M (I_=) \; = \; 10015005^2 \cdot 38567100 \cdot 211915132^2 \cdot
     886163135^4 \cdot 2054455634 \; \approx \; 2.20 \cdot 10^{83}, \]
  and $\hat{X}$ is one of these matrices. By (\ref{eq:Nd}) and (\ref{eq:bo2}),
  $\hat{X}$ can be realized in $\#( \hat{X}) \; = 308! / \bigl( (2!)^2 2.4!
  (3!)^2 (3.6!)^4 4! \bigr)^{10} \approx 1.46 \cdot 10^{549}$ ways, where we
  took some liberties by allowing non-integral entries.
  
  How much more likely is $\hat{X}$ than a matrix $X'$ which also obeys $I$
  and is the same as $\hat{X}$ except that its 5th row is
  (2,2,2,2,2,4,4,4,4,4), a slight deviation from (3,...,3)? We see that $\#(
  \hat{X}) /\#(X') = \; (2!)^5 (4!)^5 / (3!)^{10} \approx 4.21$. If row 8 is
  (2,2,2,2,2,2,2,6,8,8) instead of (3.6,...,3.6), a larger deviation, the
  likelihood of $X'$ is significantly smaller: $\; \#( \hat{X}) /\#(X') =
  (2!)^7 6! (8!)^2 / (3.6!)^{10} \approx 813.9$. Note that the units chosen
  for the $u_i$ affect the size of the absolute numbers above, as well as the
  ratios; choosing finer units increases both the numbers and the ratios
  dramatically. For example, if all the $u_i$ are multiplied by 10, the two
  likelihoods computed above become $1.8 \cdot 10^7$ and $4.2 \cdot 10^{32}$.
\end{example}

\section{Total sum and bounds on row sums}

\label{sec:bounds1}Now suppose that besides the upper bounds on the row sums
we also know the total sum $s$:
\begin{equation}
  \label{eq:bounds1} \sum_{i, j} x_{i j} = s, \hspace{2em} \tmop{and}
  \hspace{1em} \forall i, \hspace{1em} \sum_j x_{i j} \leqslant u_i .
\end{equation}
For a solution to exist, we must have $s \leqslant u_1 + \cdots + u_n .$ By
Corollary \ref{cor:H} we then have
\begin{equation}
  \label{eq:dtb} x_{i j} = \lambda_i \mu, \hspace{2em} 0 < \lambda_i \leqslant
  1.
\end{equation}
To proceed, we consider the solution to a simpler problem: given $a, b_1,
\ldots, b_n > 0$, what is the maximum entropy vector $x^{\ast}$ satisfying
$x_1 + \cdots + x_n = a$ and $\forall i, x_i \leqslant b_i$?

\subsection{The vector case}

\begin{lemma}
  \label{le:ub}The maximum-entropy vector $x^{\ast}$ satisfying $\sum_i x_i =
  a$ and $\forall i \; 0 \leqslant x_i \leqslant b_i$, where $a \leqslant b_1
  + \cdots + b_n$, is found as follows:
  \begin{enumerateroman}
    \item Arrange the $b_i$ in increasing order, and permute the $x_i$
    accordingly.
    
    \item Find the largest $j \in \{0, \ldots, n\}$ for which $b_1 + \cdots +
    b_j + (n - j) b_j \leqslant a$. Let that be $k$.
  \end{enumerateroman}
  Then $x^{\ast}_1 = b_1, x^{\ast}_2 = b_2, \ldots, x^{\ast}_k = b_k$ and
  $x^{\ast}_{k + 1} = \cdots = x^{\ast}_n = \frac{a - (b_1 + \cdots + b_k)}{n
  - k}$.
\end{lemma}

The starting point for this result is noting that if the $b_i$ are in
increasing order, there is a unique $\ell$ s.t. $b_1 + \cdots + b_{\ell} < a
\leqslant b_1 + \cdots + b_{\ell + 1}$. If so, a plausible high-entropy
solution is to set the first $\ell$ of the $x_i$ (constrained to be smallest)
equal to their upper bounds, and split the remainder of $a$, which does not
exceed $b_{\ell + 1}$, equally among the rest of the $x_i$, which are the more
loosely constrained. Lemma \ref{le:ub} refines this idea: to actually achieve
maximum entropy, only the first $k < \ell$ of the $x_i$ can be set to their
upper bounds.

The significance of $k$ is as follows. Suppose $b_1 > a / n$; then $k = 0$,
and $b_2, \ldots, b_n$ are also $> a / n$. This means that the bounds on the
$x_i$ are loose enough to allow {\tmem{complete symmetry/uniformity}}: the
{\maxent} solution is $x^{\ast}_1 = \cdots = x^{\ast}_n = a / n$. Now suppose
that $b_1 \leqslant a / n$ and $b_1 + (n - 1) b_2 > a$, in which case $k = 1$.
Then the bound $b_1$ is restrictive enough to break the symmetry: the solution
is $x_1 = b_1, x_2 = \cdots = x_n = (a - b_1) / (n - 1)$, symmetric apart from
$x_1$. So, in general, $k$ measures how many of the constraints on the
individual $x_i$ are {\tmem{informative}}, i.e. force the solution away from
the total uniformity that would have obtained if only the constraint $x_1 +
\cdots + x_n = a$ had been present. Finally, $k = n$ iff $b_1 + \cdots + b_n =
a$. In that extreme, the solution is determined completely by the upper
bounds: $x^{\ast} = (b_1, \ldots, b_n)$.

\subsection{Back to the matrix}

\label{sec:back}Returning to the solution (\ref{eq:dtb}), we proceed along the
lines of the proof of Lemma \ref{le:ub} in the Appendix. We treat the $u_i$ as
the $b_i$ of the lemma: arrange the rows of $X$ so that $u_1 \leqslant u_2
\leqslant \cdots \leqslant u_n$, and find the largest $k$ s.t.
\begin{equation}
  \label{eq:k} u_1 + \cdots + u_k + (n - k) u_k \leqslant s.
\end{equation}
It may be that $k = 0$, i.e. $u_1 > s / n$, but $k$ cannot exceed $n$. As
pointed out above, the number $k$ measures how many of the row constraints are
informative. Now consider the solution
\begin{equation}
  \label{eq:sol} \sum_j x_{1 j} = u_1, \ldots, \sum_j x_{k j} = u_k,
  \hspace{1em} \lambda_{k + 1} = \cdots = \lambda_n = 1.
\end{equation}
From (\ref{eq:dtb}), this implies that for all $j$, $x_{k + 1, j} = \cdots =
x_{n j} = \mu$. Since the sum of all $x_{i j}$ must be $s$,
\[ u_1 + \cdots + u_k + (n - k) m \mu = s, \hspace{2em} \tmop{so} \hspace{1em}
   \mu = \frac{s - (u_1 + \cdots + u_k)}{m (n - k)} . \]
Using this in (\ref{eq:sol}),
\[ \lambda_i = \frac{(n - k) u_i}{s - (u_1 + \cdots + u_k)} \hspace{2em} i =
   1, \ldots, k, \]
and we must verify that $\lambda_i \leqslant 1$. But this holds if $s > u_1 +
\cdots + u_k + (n - k) u_i$, which is true for any $i$ because of
(\ref{eq:k}).

In summary, with the rows of $X$ arranged so that $u_1 \leqslant u_2 \leqslant
\cdots \leqslant u_n$, the solution is
\begin{equation}
  \hat{x}_{i j} = \left\{ \begin{array}{ll}
    \displaystyle
    \frac{u_i}{m}, & i \leqslant k,\\
    \displaystyle
    \frac{s - (u_1 + \cdots + u_k)}{m (n - k)}, & i > k,
  \end{array} \right. \hspace{1em} j = 1, \ldots, m
  \label{eq:soltb}
\end{equation}
where $k$ is determined by (\ref{eq:k}). The non-informative $u_{k + 1},
\ldots, u_n$ do not appear.

According to (\ref{eq:soltb}), $\hat{X}$ consists of $k$ identical columns
with the structure specified in the 1st line, followed by $m - k$ identical
columns with the structure specified in the 2nd line. Within each set, the
columns are identical because we do not have any information that imposes a
distinction. We also note that if $s = \sum_i u_i$, then $k = n - 1$, and we
obtain the solution (\ref{eq:bo2}), as expected, since this value of $s$
imposes no additional constraint. If $s / n \leqslant \min_i u_i$, the matrix
is totally uniform: $\hat{x}_{i j} = s / n^2$.

Finally, the solution (\ref{eq:soltb}) translates immediately to the case
where we have bounds on the columns, instead of the rows of the matrix.

\begin{example}
  \label{ex:bounds1}We re-do Example \ref{ex:ub}, adding information on the
  total sum $s$. Here $_{} \sum_i u_i = 308$. We see from the last column of
  the table that $\#( \hat{X} \mid s)$ increases with $s$, as intuitively
  expected.
  
  \begin{table}[h]
      \begin{tabular}{|c|c|cccccccccc|c|} \hline
        $s$ & $k$ & $\hat{x}_{1 \cdot}$ &  &  &  &  & $\hat{x}_{5 \cdot}$ &  &
      &  & $\hat{x}_{10 \cdot}$ & $\log_{10} \#(\hat{X}|s)$ \\ \hline
        308 & 10 & 20 & 20 & 24 & 30 & 30 & 36 & 36 & 36 & 36 & 40 & 549.2\\
        \cline{12-12}

        \multicolumn{1}{|c|}{307} & \multicolumn{1}{c|}{9} &
        \multicolumn{1}{c}{20} & \multicolumn{1}{c}{20} &
        \multicolumn{1}{c}{24} & \multicolumn{1}{c}{30} &
        \multicolumn{1}{c}{30} & \multicolumn{1}{c}{36} &
        \multicolumn{1}{c}{36} & \multicolumn{1}{c}{36} &
        \multicolumn{1}{c}{36} & \multicolumn{1}{|c}{39} &
        \multicolumn{1}{|c|}{547.3} \\

        \multicolumn{1}{|c|}{304} & \multicolumn{1}{c|}{9} &
        \multicolumn{1}{c}{20} & \multicolumn{1}{c}{20} &
        \multicolumn{1}{c}{24} & \multicolumn{1}{c}{30} &
        \multicolumn{1}{c}{30} & \multicolumn{1}{c}{36} &
        \multicolumn{1}{c}{36} & \multicolumn{1}{c}{36} &
        \multicolumn{1}{c}{36} & \multicolumn{1}{|c}{36} &
        \multicolumn{1}{|c|}{541.8} \\  \cline{8-11}

        \multicolumn{1}{|c|}{303} & \multicolumn{1}{c|}{9} &
        \multicolumn{1}{c}{20} & \multicolumn{1}{c}{20} &
        \multicolumn{1}{c}{24} & \multicolumn{1}{c}{30} &
        \multicolumn{1}{c}{30} & \multicolumn{1}{|c}{35.8} &
        \multicolumn{1}{c}{35.8} & \multicolumn{1}{c}{35.8} &
        \multicolumn{1}{c}{35.8} & \multicolumn{1}{c}{35.8} &
        \multicolumn{1}{|c|}{539.9} \\

        \multicolumn{1}{|c|}{275} & \multicolumn{1}{c|}{5} &
        \multicolumn{1}{c}{20} & \multicolumn{1}{c}{20} &
        \multicolumn{1}{c}{24} & \multicolumn{1}{c}{30} &
        \multicolumn{1}{c}{30} & \multicolumn{1}{|c}{30.2} &
        \multicolumn{1}{c}{30.2} & \multicolumn{1}{c}{30.2} &
        \multicolumn{1}{c}{30.2} & \multicolumn{1}{c}{30.2} &
        \multicolumn{1}{|c|}{487.2} \\

        \multicolumn{1}{|c|}{274} & \multicolumn{1}{c|}{5} &
        \multicolumn{1}{c}{20} & \multicolumn{1}{c}{20} &
        \multicolumn{1}{c}{24} & \multicolumn{1}{c}{30} &
        \multicolumn{1}{c}{30} & \multicolumn{1}{|c}{30} &
        \multicolumn{1}{c}{30} & \multicolumn{1}{c}{30} &
        \multicolumn{1}{c}{30} & \multicolumn{1}{c}{30} &
        \multicolumn{1}{|c|}{485.3} \\ \cline{6-7}

        \multicolumn{1}{|c|}{273} & \multicolumn{1}{c|}{3} &
        \multicolumn{1}{c}{20} & \multicolumn{1}{c}{20} &
        \multicolumn{1}{c|}{24} & \multicolumn{1}{c}{29.86} &
        \multicolumn{1}{c}{29.86} & \multicolumn{1}{c}{29.86} &
        \multicolumn{1}{c}{29.86} & \multicolumn{1}{c}{29.86} &
        \multicolumn{1}{c}{29.86} & \multicolumn{1}{c}{29.86} &
        \multicolumn{1}{|c|}{483.4} \\

        \multicolumn{1}{|c|}{272} & \multicolumn{1}{c|}{3} &
        \multicolumn{1}{c}{20} & \multicolumn{1}{c}{20} &
        \multicolumn{1}{c|}{24} & \multicolumn{1}{c}{29.71} &
        \multicolumn{1}{c}{29.71} & \multicolumn{1}{c}{29.71} &
        \multicolumn{1}{c}{29.71} & \multicolumn{1}{c}{29.86} &
        \multicolumn{1}{c}{29.71} & \multicolumn{1}{c}{29.71} &
        \multicolumn{1}{|c|}{481.5} \\ \hline
    \end{tabular}
    \caption{Row sums of the most likely $10 \times 10$ matrix $\hat{X}$ as a
    function of $s$. Within a row, all elements are equal. The stepwise line
    inside the table indicates the $k$-boundary. The last column of the table
    is computed by (\ref{eq:Nd}).}
  \end{table}
\end{example}

\subsection{Bounds on total sum and on row sums}

\label{sec:b1}Suppose that instead of knowing the total sum as above, we have
only an upper bound $u$ on it:
\begin{equation}
  \sum_{i, j} x_{i j} \leqslant u, \hspace{2em} \tmop{and} \hspace{1em}
  \forall i, \hspace{1em} \sum_j x_{i j} \leqslant u_i .
\end{equation}
What has already been said in this section suffices to solve this problem
also. First, if $u > \sum_i u_i$, then this constraint is immaterial and we
have the problem of {\S}\ref{sec:bounds2}, whose solution is given by
(\ref{eq:bo2}). So we are left with the case $u \leqslant \sum_i u_i$. Suppose
that we pick a value $s < u$ for the total demand, and then find $\hat{X}$ as
in {\S}\ref{sec:back}. Example \ref{ex:bounds1} showed that $\#( \hat{X} \mid
s)$ increases as $s$ increases, suggesting that we should reduce to the
problem (\ref{eq:bounds1}) with $s = u$. Indeed, Lemma \ref{le:ub2}
establishes this formally.

This is the first case where ``most likely'' is not equivalent to ``having
maximum entropy''. However, we see that there is still a strong and simple
connection: the most likely matrix is the {\maxent} matrix with the largest
total sum allowed by the constraints.

\section{Bounds on row and column sums}

\label{sec:boundsrc}Here we consider the situation where our information $I$
consists just of upper bounds on both the row and column sums of the matrix:
\[ \sum_j x_{i j} \leqslant u_i, \hspace{1em} \sum_i x_{i j} \leqslant v_j,
   \hspace{2em} i, j = 1, \ldots, n. \]
The number of realizations of a matrix subject to this information is given by
expression (\ref{eq:Nd}), except in this case the total sum $s$ is not known
and has to be substituted by $\sum_{i, j} x_{i j}$. If we use the first two
terms of (\ref{eq:stirling}) to approximate the log of
\[ \#(X|I) \; = \; \left(\begin{array}{c}
     x_{11} + \cdots + x_{n m}\\
     x_{11}, \ldots, x_{1 m}, x_{21}, \ldots, x_{2 m}, \ldots, x_{n 1},
     \ldots, x_{n m}
   \end{array}\right), \]
we find that it is given by the ``entropy difference'' function
\begin{equation}
  \begin{array}{lll}
    G (X) & = & \Bigl( \sum_{i, j} x_{i j} \Bigr) \ln \Bigl( \sum_{i, j} x_{i
    j} \Bigr) - \sum_{i, j} x_{i j} - \sum_{i, j} (x_{i j} \ln x_{i j} - x_{i
    j})\\
    & = & \Bigl( \sum_{i, j} x_{i j} \Bigr) \ln \Bigl( \sum_{i, j} x_{i j}
    \Bigr) - \sum_{i, j} x_{i j} \ln x_{i j} .
  \end{array} \label{eq:G}
\end{equation}
(When $I$ includes the value of $\sum_{i, j} x_{i j}$, maximizing $G (X)$
subject to $I$ is equivalent to maximizing $H (X)$ subject to $I$.)
Proposition \ref{prop:G} in the Appendix shows that $G (X)$ is concave over
the domain $x_{i j} > 0$. And by Corollary \ref{cor:G}, the elements of
$\hat{X}$ have the form
\begin{equation}
  \hat{x}_{i j} \; = \; \Bigl( \sum_{k, l} \hat{x}_{k l} \Bigr) \lambda_i
  \mu_j, \hspace{1em} \lambda_i, \mu_j \in (0, 1] . \label{eq:xrcb}
\end{equation}
Given the above, we note that there are two cases to consider w.r.t. to the
bounds:
\begin{enumerate}
  \item All rows sum to their bounds, and all columns sum to their bounds.
  
  \item At least one row or one column sums to less than its bound.
\end{enumerate}
Case 1 is possible only when $\sum_i u_i = \sum_j v_j$. If so, the solution
s.t. $\forall i, \sum_j x_{i j} = u_i$ and $\forall j, \sum_i x_{i j} = v_j$
has been discussed in {\S}\ref{sec:rcs}. Thus we need only consider case 2. We
can establish the following property of $\hat{X}$:

\begin{proposition}
  \label{prop:rcb}The matrix $\hat{X}$ is s.t. for any $i, j$ pair, either row
  $i$ sums to $u_i$, or column $j$ sums to $v_j$. That is, there can be no
  pair $i, j$ s.t. row $i$ sums to $< u_i$ and column $j$ sums to $< v_j$.
\end{proposition}

By virtue of Proposition \ref{prop:rcb}, if \tmtextit{one} column of $\hat{X}$
sums to less than its bound, then \tmtextit{all} rows must sum to their
bounds. The situation is symmetric w.r.t. rows and columns, so we will analyze
just the column case, where one or more columns sum to less than their bounds.

So suppose that columns $1, \ldots, k$ sum to their bounds, while columns $k +
1, \ldots, m$ sum to less than their bounds, with $0 \leqslant k < m$. Then we
must have $v = \sum_j v_j > \sum_i u_i = u$. By Corollary \ref{cor:G}, $\mu_{k
+ 1} = \cdots = \mu_m = 1$ in (\ref{eq:xrcb}). Also, as pointed out above, all
rows must sum to their bounds, which implies that $\sum_{k, l} x_{k l} = u$.

If we consider the columns, (\ref{eq:xrcb}) says that $x_{i j} = u \lambda_i
\mu_j$ for $j \leqslant k$, and $x_{i j} = u \lambda_i$ for $j > k$. Adding
these by sides over $i$ we obtain
\begin{equation}
  v_j = u \lambda \mu_j, \hspace{1em} j \leqslant k \hspace{2em} \text{and}
  \hspace{2em} v_{k + 1} > u \lambda, \ldots, v_m > u \lambda, \label{eq:rcb1}
\end{equation}
where $\lambda$ is the sum of the $\lambda_i$. Further, if we add all the
columns, $v_1 + \cdots + v_k + u \lambda + \cdots + u \lambda = u$, whence
\begin{equation}
  \lambda = \frac{u - (v_1 + \cdots + v_k)}{(m - k) u} . \label{eq:rcb2}
\end{equation}
($\lambda = 1 / m$ if $k = 0$, i.e. if all columns sum to less than their
bounds.) Turning to the rows, we have $u \lambda_1 \mu = u_1, \ldots, u
\lambda_n \mu = u_n$, where $\mu$ is the sum of the $\mu_j$. Thus
\begin{equation}
  \lambda_i \mu = \frac{u_i}{u} = r_i, \hspace{2em} \text{and} \hspace{2em}
  \lambda \mu = 1. \label{eq:rcb3}
\end{equation}
We can now determine all the $\lambda_i$ and $\mu_j$: from (\ref{eq:rcb3}) and
(\ref{eq:rcb1}),
\begin{equation}
  \lambda_i = \lambda r_i \hspace{1em} \text{and} \hspace{1em} \mu_j = \left\{
  \begin{array}{ll}
    (1 / \lambda) v_j / u, & j \leqslant k.\\
    1, & j > k.
  \end{array} \right. \label{eq:rcb4}
\end{equation}
Note that neither $\lambda_i$ nor $\mu_j$ depend on $v_{k + 1}, \ldots, v_m$.
This means that we can take these bounds to be as large as we please, e.g.
$\infty$, thus handling the case where no upper bound is specified for some of
these columns.

It remains to take care of the fact that (\ref{eq:xrcb}) requires $\lambda_i,
\mu_j \leqslant 1$. It is easy to verify this for $\lambda_i$: it is the
product of two factors, both $< 1$. The condition $\mu_j \leqslant 1$ is
equivalent to $(m - k) v_j \leqslant u - (v_1 + \cdots + v_k)$ for $j
\leqslant k$. The inequalities in (\ref{eq:rcb1}) impose another condition on
$k$: $(m - k) \min (v_{k + 1}, \ldots, v_m) > u - (v_1 + \cdots + v_k)$.
Taking these two conditions together we see that $k$ and the column bounds
$v_j$ must satisfy
\[ \max (v_1, \ldots, v_k) \; \leqslant \; \frac{u - (v_1 + \cdots + v_k)}{m -
   k} \; < \; \min (v_{k + 1}, \ldots, v_m), \]
where $0 \leqslant k < m$. Assume that the $v_j$ are in increasing
order{\footnote{If the matrix is a contingency table, simply rearrange the
columns. If it refers to $n$ nodes, re-label the nodes.}}. Then this condition
becomes
\begin{equation}
  \; v_k \leqslant \; \frac{u - (v_1 + \cdots + v_k)}{m - k} \; < \; v_{k +
  1}, \hspace{2em} 0 \leqslant k < m. \label{eq:rcb5}
\end{equation}
The following result establishes the existence of a $k$ satisfying
(\ref{eq:rcb5}):

\begin{proposition}
  \label{prop:rcb2}Let $v_1 \leqslant v_2 \leqslant \cdots \leqslant v_m$, $u
  < v$, and $v_1 \leqslant u / m$. Then there is a unique $k \in \{1, \ldots,
  m - 1\}$ s.t.
  \[ (m - k) v_k \; \leqslant \; u - (v_1 + \cdots + v_k) \; < \; (m - k) v_{k
     + 1} . \]
  If $v_1 > u / m$, then $k = 0.$
\end{proposition}

Finally, if $u < v$, by (\ref{eq:xrcb}), (\ref{eq:rcb2}), and (\ref{eq:rcb4}),
the elements of $\hat{X}$ are
\begin{equation}
  \hat{x}_{i j} \; = \; \left\{ \begin{array}{ll}
    \displaystyle
    \frac{u_i v_j}{u}, & j \leqslant k,\\
    \displaystyle
    \frac{u - (v_1 + \cdots + v_k)}{m - k}  \frac{u_i}{u}, & j > k,
  \end{array} \right. 
  \label{eq:rcbfinal} \hspace{1em} i = 1, \ldots, n
\end{equation}
where $k$ is given by Proposition \ref{prop:rcb2}. We note that $k$ is the
number of \tmtextit{informative} column constraints, in the sense that the
solution depends on $v_1, \ldots, v_k$ but not on $v_{k + 1}, \ldots$ (similar
to the $k$ in Lemma \ref{le:ub}). In fact, some of $v_{k + 1}, \ldots, v_m$
may be infinite, i.e. there may be no upper bounds on some of columns $k + 1,
\ldots, m$.

The reader may want to compare (\ref{eq:rcbfinal}) with the result
(\ref{eq:gm}) for the model of {\S}\ref{sec:rcs}. The comparison shows that
even though for the problem we just solved ``most likely'' is not directly
equivalent to ``having maximum entropy'', there is still a straightforward
connection as we also saw in {\S}\ref{sec:b1}.

\section{Bounds on individual elements}

\label{sec:bindiv}We first point out that whereas bounds on individual matrix
elements provide the utmost flexibility in expressing constraints, they can
have unintended consequences. Then we look at the most likely matrix subject
to bounds on the row sums and on individual elements.

\subsection{Expressive power and consistency}

\subsubsection{Expressive power}

Consider finding the most likely matrix $\hat{X}$ subject just to the
constraints
\[ \forall i, j \hspace{1em} x_{i j} \leqslant w_{i j}, \]
where $W$ is a given $n \times m$ matrix in $\mathbb{N}$. Then it is easy to
see by an argument similar to that of {\S}\ref{sec:bounds2} that $\hat{X}$ has
elements $\hat{x}_{i j} = w_{i j}$. Thus the information $W$ suffices to
specify {\tmem{any}} possible most likely matrix. Conversely, to be able to
specify an arbitrary matrix, information on every matrix element is necessary;
$W$ is one form of such information.

\subsubsection{\label{sec:cons}Consistency}

Imposing $w$-constraints that are satisfied with equality requires that the
$w$- and $u$-constraints together satisfy certain conditions if the matrix
$\hat{X}$ is not to exhibit surprising behavior. For example, suppose we are
trying to determine a $3 \times 3$ matrix with row/column sums $u_1, u_2, u_3$
and s.t. $x_{11} = 0$, as shown in Fig. \ref{fig:w}, left. Then we must have
$x_{12} + x_{13} = u_1$ from row 1, and $(x_{22} + x_{23} + x_{32} + x_{33}) +
x_{12} + x_{13} = u_2 + u_3$ from columns 2 and 3. It follows that if $u_1$ is
not strictly less than $u_2 + u_3$, then $x_{22} + x_{23} + x_{32} + x_{33} =
0$, which means that all these elements are 0. So with certain $u_1, u_2,
u_3$, $x_{11} = 0$ may force other elements of the matrix to be 0 as well.
This does not happen without the requirement $x_{11} = 0$: we know from
{\S}\ref{sec:rcs} that for any $u_1, u_2, u_3$ there is a $\hat{X}$ with all
elements non-zero.

\begin{figure}[h]
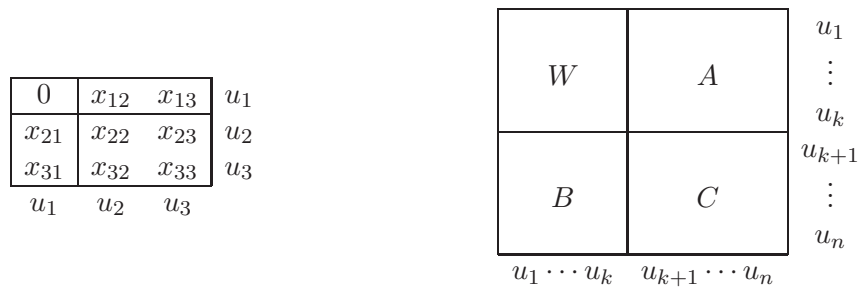

  \[ \begin{array}{|c|cc|c} \cline{1-3}
       0 & x_{12} & x_{13} & u_1\\ \cline{1-3}
       x_{21} & x_{22} & x_{23} & u_2\\
       x_{31} & x_{32} & x_{33} & u_3\\ \cline{1-3}
       \multicolumn{1}{c}{u_1} & \multicolumn{1}{c}{u_2} &
       \multicolumn{1}{c}{u_3} & 
     \end{array} \hspace{8em}
     \begin{array}{|c|c|c} \cline{1-2}
       &  & u_1 \\
       W & A & \vdots \\
       &  & u_k \\ \cline{1-2}
       &  & u_{k + 1}\\
       B & C & \vdots\\
       &  & u_n \\ \cline{1-2}
       \multicolumn{1}{c}{u_1 \cdots u_k} &
       \multicolumn{1}{c}{u_{k+1} \cdots u_n} & 
     \end{array}
   \]
  \caption{\label{fig:w}Matrices with some elements fixed by $w$-constraints.}
\end{figure}

More generally, suppose we have a constraint that forces a certain $k \times
k$ square submatrix of $X$ to equal a matrix $W$. In the simplest case let $W$
be in the upper left-hand corner of $X$ as shown in Fig. \ref{fig:w}, right.
Then we have
\[ \Sigma_W + \Sigma_A = u_1 + \cdots + u_k, \hspace{2em} \Sigma_A + \Sigma_C
   = u_{k + 1} + \cdots + u_n, \]
assuming that $\Sigma_W < u_1 + \cdots + u_k$. It can be seen that unless $u$
and $W$ are s.t. $u_1 + \cdots + u_k - \Sigma_W < u_{k + 1} + \cdots + u_n$,
we must have $\Sigma_C = 0$, which would force the entire submatrix $C$ to be
0.

If $W$ is an arbitrary submatrix, let its rows and columns correspond to a set
$I$ of indices. Then the condition that must be satisfied so that $C$ is not
forced to 0 can be written as
\begin{equation}
  u_I < \frac{u}{2} + \frac{w_{I I}}{2}, \label{eq:cons}
\end{equation}
where the subscripts indicate summation over the set. In terms of a traffic
matrix, this condition says that the traffic originating in the set $I$ must
be less than half of the total, plus half of the traffic originating in $I$
and terminating in $I$. As an example, suppose we require that there is no
traffic among the locations in $I$; then (\ref{eq:cons}) says that the traffic
that leaves $I$ cannot be more than half of the total traffic.

Related to the above, there is also a necessary and sufficient condition for
the existence of a non-negative matrix with specified row and column sums and
an arbitrary subset of elements specified to be 0: see Theorems 3.10 and 3.12
in Ch. 4 of {\cite{mat}}; see also {\S}3.6 of {\cite{grav}}.

\subsection{Bounds on row sums and on individual elements}

\label{sec:bounds3}Suppose we know the same bounds on the row sums as in 3,
but, in addition, we have a bound on the size of each individual element:
\begin{equation}
  \forall i \hspace{1em} \sum_j x_{i j} \leqslant u_i, \hspace{2em} \tmop{and}
  \hspace{2em} \forall i, j \hspace{1em} x_{i j} \leqslant w_{i j} .
  \label{eq:brsie}
\end{equation}
This problem is easy to solve because the constraints (\ref{eq:brsie}) are
separable, so each row of the most likely matrix $\hat{X}$ can be found
{\tmem{independently}} of all the other rows. Fixing a particular row $i$,
denote the $x_{i j}$ by $x_1, \ldots, x_m$, $u_i$ by $a$, and the $w_{i j}$ by
$b_1, \ldots, b_m$. Then we have the problem of finding the most likely vector
$x^{\ast}$ that satisfies
\begin{equation}
  x_1 + \cdots + x_m \leqslant a, \hspace{2em} x_1 \leqslant b_1, \ldots, x_m
  \leqslant b_m, \hspace{2em} a, b_i \in \mathbb{N}.
\end{equation}
The solution to this problem is as follows. If $a > b_1 + \cdots + b_m$,
$x^{\ast}$ is simply $(b_1, \ldots, b_m)$. If $a \leqslant b_1 + \cdots +
b_m$, then $x^{\ast}$ is found by replacing the inequality with an equality
and reducing to the problem solved in \ref{sec:bounds1}. The formal details
are given in Lemma \ref{le:ub2} in the Appendix.

\section{Symmetric information}

\label{sec:symm}We now investigate some types of constraints that we have not
looked at so far, but under the additional assumption that these constraints
or information are {\tmem{symmetric}} w.r.t. rows and columns. By necessity,
the matrices are $n \times n$ square. Here is one motivation for considering
symmetry. Suppose we are designing a ``backbone'' type, i.e. high capacity and
geographically extensive, communications network connecting a set of $n$
locations. The {\tmem{transmission facilitie}}s in such a network typically
have the same capacity in each direction. To understand what capacities are
needed for the links, we need an estimate of the traffic matrix. Given the
symmetry of the capacities, we may assume, for example, that the total
incoming and outgoing traffic for a given node are equal. The same
considerations apply to a network of roads connecting a set of cities. Thus
the symmetry of capacities allows us to act as if the traffic matrix were
symmetric.

These considerations aside, the symmetric information allows us to go farther
toward analytical solutions than would be possible otherwise. One of the
questions we investigate via the analytical forms is the effect of fixing some
elements on the ``product of independent factors'' structure of the {\maxent}
matrix.

\subsection{Total sum and bounds on row and column sums}

Here the sum of row $i$ is bounded by $u_i$, and so is the sum of column $i$.
By Corollary \ref{cor:H}, the matrix elements are of the form
\[ x_{i j} = \lambda_i \mu_j \nu, \hspace{2em} \lambda_i, \mu_j \in (0, 1], \]
where $\lambda_i, \mu_i$ correspond to the row and column constraints
respectively, and $\nu$ to the constraint on the total sum. We will show that
the solution is essentially the same as that obtained in {\S}\ref{sec:bounds1}
for the non-symmetric, row-only case. So define $k$ by (\ref{eq:k}), and
consider the solution
\begin{enumerate}
  \item Constraints $1, \ldots, k$ are satisfied as equalities for both rows
  and columns (so we must have $\lambda_1, \ldots, \lambda_k \leqslant 1$ and
  $\mu_1, \ldots, \mu_k \leqslant 1$), and
  
  \item $\lambda_{k + 1} = \cdots = \lambda_n = 1$, and $\mu_{k + 1} = \cdots
  = \mu_n = 1$.
\end{enumerate}
It follows that the matrix must look like
\[ \left[ \begin{array}{ll}
     {}[\lambda_i \mu_j \nu]_{k \times k} & [\lambda_i \nu]_{k \times n - k}\\
     {}[\mu_j \nu]_{n - k \times k} & [\nu]_{n - k \times n - k}
   \end{array} \right] \; = \; \left[ \begin{array}{ll}
     A & B\\
     C & D
   \end{array} \right] . \]
Note that rows $k + 1, \ldots, n$ are identical, and so are columns $k + 1,
\ldots, n$. Let $\Sigma_A, \ldots, \Sigma_D$ denote the sums of the elements
of the submatrices. Clearly,
\[ \Sigma_A + \Sigma_B = u_1 + \cdots + u_k, \hspace{1em} \Sigma_A + \Sigma_C
   = u_1 + \cdots + u_k, \hspace{1em} \Sigma_A + \Sigma_B + \Sigma_C +
   \Sigma_D = s. \]
Therefore $\Sigma_B = \Sigma_C$ and $u_1 + \cdots + u_k + \Sigma_B + \Sigma_D
= s$. Substituting for the elements of $B$ and $D$, we find that
\[ \nu = \frac{s - (u_1 + \cdots + u_k)}{(n - k) (\lambda_1 + \cdots +
   \lambda_k + n - k)} . \]
And from $\Sigma_B = \Sigma_C$ it follows that $\lambda_1 + \cdots + \lambda_k
= \mu_1 + \cdots + \mu_k$. Now the constraint on row $i < k$ is $\lambda_i
(\lambda_1 + \cdots + \lambda_k + n - k) \nu = u_i$. Using the expression for
$\nu$ in this we find that $\lambda_i = (n - k) u_i / (s - (u_1 + \cdots +
u_k))$, and this is $< 1$ as required. Similarly, from the constraint for
column $j < k$ we find that $\mu_j = (n - k) u_j / (s - (u_1 + \cdots +
u_k))$. We see that $\mu_j = \lambda_j$. Therefore we need only deal with the
$\lambda_j$. Substituting the values of the $\lambda_i$ in the expression for
$\nu$, we finally arrive at the solution
\begin{equation}
  \hat{x}_{i j} \; = \; \left\{ \begin{array}{ll}
    \displaystyle
    \frac{u_i u_j}{s}, & (i, j) \in A,\\
    \displaystyle
    \frac{(s - (u_1 + \cdots + u_k)) u_i}{(n - k) s}, & (i, j) \in B,\\
    \displaystyle
    \frac{(s - (u_1 + \cdots + u_k)) u_j}{(n - k) s}, & (i, j) \in C,\\
    \displaystyle
    \frac{(s - (u_1 + \cdots + u_k))^2}{(n - k)^2 s}, & (i, j) \in D,
  \end{array} \right.
\end{equation}
where $k$ is defined by (\ref{eq:k}). This solution is symmetric, and the
reader can verify that it satisfies all the constraints. The $A, B, C$
matrices have the gravity form, the $B$ and $C$ matrices are the transpose of
one another, and the $D$ matrix is constant. Finally, note that we did not
assume the symmetry in the solution, it followed as a consequence of the
symmetric information.

\subsection{Given row and column sums, partially fixed diagonal}

\label{sec:fixd}Assume that the sum of row and column $i$ is $u_i$, and that
the first $m \leqslant n$ of the diagonal elements are fixed to be 0. Then,
with $s = u_1 + \cdots + u_n$ still denoting the total sum, the matrix
elements other than the first $m$ on the diagonal must be given by
\[ \hat{x}_{i j} = s \lambda_i \mu_j, \hspace{1em} \lambda_i, \mu_j > 0. \]
Including the factor $s$ is a convenience, as will become clear. For the above
solution to be possible the consistency condition (\ref{eq:cons}) must be
satisfied: each $u_i$ must be strictly less than half of $s$. We shall assume
this to be the case. If $\lambda$ is the sum of the $\lambda_i$ and $\mu$ that
of the $\mu_j$, and $r_i = u_i / s$, the row and column constraints can be
written as
\begin{equation}
  \begin{array}{llll}
    \lambda_i (\mu - \mu_i) = r_i, & i \leqslant m, \hspace{1em} & \lambda_i
    \mu = r_i, & i > m,\\
    \mu_j (\lambda - \lambda_j) = r_j, & j \leqslant m, \hspace{1em} & \mu_j
    \lambda = r_j, & j > m.
  \end{array} \label{eq:lsys1}
\end{equation}
Here $r_1 + \cdots + r_n = 1$ and $r_i < 1 / 2$ for all $i$. Noting that
(\ref{eq:lsys1}) is unchanged if we exchange the $\lambda_i$ and the $\mu_j$
leads us to consider a solution with $\mu_i = \lambda_i$. Then
(\ref{eq:lsys1}) reduces to
\begin{equation}
  \lambda_i (\lambda - \lambda_i) = r_i, \hspace{1em} i \leqslant m,
  \hspace{2em} \lambda_i \lambda = r_i, \hspace{1em} i > m. \label{eq:lsys2}
\end{equation}
(\ref{eq:lsys2}) implies that for $i \leqslant m$ we have $\lambda_i =
(\lambda \pm \sqrt{\lambda^2 - 4 r_i}) / 2$, whereas for $i > m$, $\lambda_i =
r_i / \lambda$. Suppose we pick the root with the ``$-$'' for $i = 1, \ldots,
m$. Adding the expressions for the $\lambda_i$ by sides and dividing both
sides of the result by $\lambda \neq 0$ we see that $\lambda$ must satisfy the
equation
\begin{equation}
  \label{eq:lambda} \sqrt{1 - 4 r_1 / \lambda^2} + \cdots + \sqrt{1 - 4 r_m /
  \lambda^2} \; - \; 2 \frac{r_{m + 1} + \cdots + r_n}{\lambda^2} \; = \; m -
  2.
\end{equation}
An exact analytical solution of (\ref{eq:lambda}) is impractical, but we can
find an approximation. To begin with, we observe that the l.h.s. of
(\ref{eq:lambda}) is a monotone increasing function of $\lambda$ so the root
of (\ref{eq:lambda}) is unique{\footnote{This is also true because the
solution to our strictly concave maximization problem is unique.}}. Second, at
the expense of restricting the $r_i$ somewhat, we can localize the root:

\begin{proposition}
  \label{prop:eqlambda}Suppose that each of $r_1, \ldots, r_n$ is in $(0, 1 /
  3)$, and $r_1 + \cdots + r_n = 1$. Then for any $n \geqslant 3$ and any $m
  \leqslant n$, equation (\ref{eq:lambda}) has a root in the interval $(2
  \sqrt{r_{\max}}, 4 / 3)$, where $r_{\max}$ is the largest of the $r_i$.
\end{proposition}

To see the necessity for some additional restriction on the $r_i$, suppose
that $m = n$ and that we extend $(0, 1 / 3)$ to $(0, 1 / 2)$. Then consider
the set $r_1 = \frac{1}{2}$ and $r_2 = \cdots = r_n = \frac{1}{2 (n - 1)}$; it
can be seen that (\ref{eq:lambda}) has no solution in $( \sqrt{2}, \infty)$.

In terms of the root $\lambda$ of (\ref{eq:lambda}) the final solution is
\begin{equation}
  \hat{x}_{i j} \; = \left\{
  \begin{array}{ll}
    \displaystyle
    \frac{s \lambda^2}{4}  \Bigl( 1 - \sqrt{1 - 4 r_i / \lambda^2} \Bigr) 
    \Bigl( 1 - \sqrt{1 - 4 r_j / \lambda^2}  \Bigr), & i, j \leqslant m
    \tmop{and} i \neq j,\\[1.5ex]
    \displaystyle
    \frac{s r_i}{2}  \Bigl( 1 - \sqrt{1 - 4 r_j / \lambda^2}  \Bigr), & i > m,
    j \leqslant m,\\[1.5ex]
    \displaystyle
    \frac{s r_j}{2}  \Bigl( 1 - \sqrt{1 - 4 r_i / \lambda^2}  \Bigr), & i
    \leqslant m, j > m,\\
    \displaystyle
    s r_i r_j / \lambda^2, & i, j > m.
  \end{array}
  \right. \label{eq:zdiag}
\end{equation}
We see that the $\hat{x}_{i j}$ for $i, j \leqslant m$ have a product form,
but, in general, the factors are \tmtextit{not} independent. We know from
{\S}\ref{sec:rcs} that irrespective of symmetry, the dependence disappears if
we don't fix any diagonal elements. Fixing these elements imposes a global
dependence as we saw in {\S}\ref{sec:bindiv}.

\begin{example}
  \label{ex:72}We consider the two extreme cases $m = n$ and $m = 1$. In
  addition, suppose that all the $u_i$ are equal.
  
  First let $m = n$. Then all $r_i$ are $1 / n$ and $\lambda = \sqrt{n / (n -
  1)}$. From the first line of (\ref{eq:zdiag}) the matrix $\hat{X}$ has the
  form
  \[ \frac{s}{n (n - 1)}  \left(\begin{array}{ccccc}
       0 & 1 & 1 & \ldots & 1\\
       1 & 0 & 1 & \ldots & 1\\
       \vdots & \vdots & \vdots & \vdots & \vdots\\
       1 & 1 & 1 & \ldots & 0
     \end{array}\right) . \]
  Compare this with the case where the diagonal is not fixed to 0, and the
  solution is $\hat{x}_{i j} = s / n^2$.
  
  Now let $m = 1$. This is the simplest possible case: we have a square matrix
  with all row and column sums equal, except that the single element $d_{11}$
  is fixed to be 0. From (\ref{eq:lambda}) we find $\lambda = (n - 1) /
  \sqrt{n (n - 2)}$. From the last three lines of (\ref{eq:zdiag}) we see that
  now $\hat{X}$ is
  \[ \frac{s}{n (n - 1)}  \left(\begin{array}{ccccc}
       0 & 1 & 1 & \ldots & 1\\
       1 & \frac{n - 2}{n - 1} & \frac{n - 2}{n - 1} & \ldots & \frac{n - 2}{n
       - 1}\\
       \vdots & \vdots & \vdots & \ldots & \vdots\\
       1 & \frac{n - 2}{n - 1} & \frac{n - 2}{n - 1} & \ldots & \frac{n - 2}{n
       - 1}
     \end{array}\right) . \]
\end{example}

\begin{example}
  \label{ex:73}Now consider the case $m = n$, but with the $r_i$ arbitrary. We
  obtain an analytical aproximation to the solution. If we let $\xi = 4 /
  \lambda^2$, (\ref{eq:lambda}) becomes
  \begin{equation}
    \sqrt{1 - r_1 \xi} + \cdots + \sqrt{1 - r_n \xi} = \; n - 2, \hspace{2em}
    \xi \in (9 / 4, 1 / r_{\max}), \label{eq:xi}
  \end{equation}
  where the lower bound on $\xi$ comes from Proposition \ref{prop:eqlambda}.
  This equation has the form $f (\xi) = c$, and if $\xi_0$ is an approximation
  to its solution, the reversion technique in Ch. 1 of {\cite{Hen}} can be
  used to find the following power series for $\xi$: with $\rho_i = \sqrt{1 -
  r_i \xi_0}$ and $\delta = f (\xi_0) - c = \rho_1 + \cdots + \rho_n - n + 2$,
  \begin{equation}
    \xi \; = \; \xi_0 + \frac{2}{r_1 / \rho_1 + \cdots + r_n / \rho_n} \delta
    - \frac{r_1^2 / \rho_1^3 + \cdots + r_n^2 / \rho_n^3}{(r_1 / \rho_1 +
    \cdots + r_n / \rho_n)^3} \delta^2 - \ldots \label{eq:ps}
  \end{equation}
  It is known that this series converges, and it can be shown that $\delta <
  1$ for any $\xi_0 \in (9 / 4, 1 / r_{\max})${\footnote{Note that $\rho_i < 1
  - 1 / 2 r_i \xi_0$, so $\rho_1 + \cdots + \rho_n < n - \xi_0 / 2$.}}. By
  (\ref{eq:zdiag}) the non-diagonal matrix elements are given by
  $\frac{s}{\xi} (1 - \sqrt{1 - r_i \xi}) (1 - \sqrt{1 - r_j \xi})$, and
  (\ref{eq:ps}) lets us find power series expansions for them in terms of
  $\delta$. We do not show these series here, but the expansions to first
  order result in manageable expressions. The accuracy of the expansions
  remains to be investigated.
  
  Now consider a numerical example with $u = (40, 20, 30, 40)$. We have $r =
  \bigl( \frac{4}{13}, \frac{2}{13}, \frac{3}{13}, \frac{4}{13} \bigr)$, so
  solving (\ref{eq:xi}) we find $\xi \approx 2.88018 \in \bigl( \frac{9}{4},
  \frac{13}{4} \bigr)$. If we take $\xi_0 = 9 / 4$, (\ref{eq:ps}) gives $\xi
  \approx 2.25 + 0.749023 - 0.112889 = 2.8861$. Then the form $\frac{s}{\xi}
  (1 - \sqrt{1 - r_i \xi}) (1 - \sqrt{1 - r_j \xi})$ yields
  \[ \hat{X} \; = \; \left( \begin{array}{cccc}
       0 & 7.59 & 12.59 & 19.82\\
       7.59 & 0 & 4.82 & 7.59\\
       12.59 & 4.82 & 0 & 12.59\\
       19.82 & 7.59 & 12.59 & 0
     \end{array} \right), \hspace{1em} \text{vs.} \hspace{1em}
     \left(\begin{array}{cccc}
       12.31 & 6.15 & 9.23 & 12.31\\
       6.15 & 3.08 & 4.62 & 6.15\\
       9.23 & 4.62 & 6.92 & 9.23\\
       12.31 & 6.15 & 9.23 & 12.31
     \end{array}\right), \]
  the {\maxent} matrix without the 0-diagonal constraint, whose elements are
  simply $s r_i r_j$. As we also saw in Example \ref{ex:72}, the result of
  fixing the diagonal to 0 cannot be regarded as a (small) perturbation of the
  $s r_i r_j$ form.
\end{example}

\paragraph{Generalization}

(a) The solution (\ref{eq:zdiag}) is valid also when the $u_i$ are upper
bounds on the row sums, instead of specifying their values. In that case
Corollary \ref{cor:H} requires that $\lambda_i \leqslant 1$, which is true if
\[ \forall i \hspace{1em} 2 \sqrt{r_i} \leqslant \lambda \leqslant 2
   \hspace{2em} \text{or} \hspace{2em} \lambda > 2. \]
But this holds by virtue of Proposition \ref{prop:eqlambda}.

(b) The diagonal elements can be set to arbitrary values $w_{11}, \ldots, w_{n
n}$, if the $r_i$ are re-defined as $(u_i - w_{i i}) / s$. This actually
requires a slight extension of Proposition \ref{prop:eqlambda}; see
Proposition \ref{prop:eqlambda2} below. And it can be verified that if we set
$w_{i i} = u_i^2 / s$ we get the expected solution $\hat{x}_{i j} = s r_i
r_j$.

\subsection{3-dimensional matrices with fixed diagonal}

The development of {\S}\ref{sec:fixd} can be extended to 3-dimensional
matrices. These can be thought of as contingency tables involving elements
with 3 attributes, or as trip matrices where a trip is characterized by an
origin and a destination as in the 2-dimensional case, and, in addition, by a
class of vehicle, say, or as traffic matrices where traffic flows have
origins, destinations, and a size class, such as ``small'', ``medium'',
``large''. Whatever the three attributes, we will index them by $i, j, k$. We
will consider the case where the whole diagonal is 0 and the available
information is the sums over all $(i, k)$ sections and all $(j, k)$ sections
of the matrix:
\[ \forall i \hspace{1em} \sum_{j \neq i} x_{i j k} = u_{i k}, \hspace{2em}
   \forall j \hspace{1em} \sum_{i \neq j} x_{i j k} = v_{j k} . \]
In the case of a traffic matrix for example, this means that we know the total
number of flows originating at $i$ and of size class $k$, and the total number
ending at $j$ of size class $k$. The matrix elements will then be
\[ x_{i j k} = s \lambda_{i k} \mu_{j k} \hspace{1em} \text{for } i \neq j,
   \hspace{1em} \text{and $0$ otherwise}, \]
where the $\lambda_{i k}$ and $\mu_{j k}$ are s.t.
\[ \forall i \hspace{1em} s \sum_{j \neq i} \lambda_{i k} \mu_{j k} = u_{i k},
   \hspace{2em} \forall j \hspace{1em} s \sum_{i \neq j} \lambda_{i k} \mu_{j
   k} = v_{j k} . \]
Now let this information be symmetric w.r.t $i$ and $j$, i.e. $v_{i k} = u_{i
k}$. Further, define $r_{i k} = u_{i k} / s$. Then the above constraints can
be written as
\[ \forall i \hspace{1em} \lambda_{i k} (\mu_{. k} - \mu_{i k}) = r_{i k},
   \hspace{2em} \forall j \hspace{1em} \mu_{j k} (\lambda_{. k} - \lambda_{j
   k}) = r_{j k}, \]
where the dot indicates summation over the corresponding index. Since the
index $j$ in the second set of constraints could have equally well been
written $i$, we are led to consider a solution with $\mu_{i k} = \lambda_{i
k}$ and the single set of constraints
\[ \forall i \hspace{1em} \lambda_{i k} (\lambda_{. k} - \lambda_{i k}) = r_{i
   k} . \]
Proceeding as we did after (\ref{eq:lsys2}), $\lambda_{i k} = \bigl(
\lambda_{. k} - \sqrt{\lambda_{. k}^2 - 4 r_{i k}} \bigr) / 2$. Adding these
over $i$ and setting $\xi_k = 4 / \lambda_{. k}^2$, we arrive at a
generalization of (\ref{eq:xi}):
\[ \forall k \hspace{1em} \sqrt{1 - r_{1 k} \xi_k} + \cdots + \sqrt{1 - r_{n
   k} \xi_k} \; = \; n - 2. \]
This is completely analogous to what we found in Example \ref{ex:73}, except
that here we have one equation for each of the $\xi_k$. The final expression
for the elements of the matrix is
\[ x_{i j k} \; = \; \frac{4}{\xi_k}  \bigl( 1 - \sqrt{1 - r_{i k} \xi_k}
   \bigr)  \bigl( 1 - \sqrt{1 - r_{j k} \xi_k} \bigr) \hspace{1em} \text{for }
   i \neq j, \hspace{1em} \text{and $0$ otherwise} . \]
Note that the matrix sections corresponding to different values of $k$ are
independent of one another. The above development generalizes to the case
where only the first $m < n$ of the diagonal elements are fixed, and in the
other ways discussed in {\S}\ref{sec:fixd}.

\subsection{Given row and column sums, fixed diagonal blocks}

We generalize the development of {\S}\ref{sec:fixd} to equality constraints
expressed by a block-diagonal matrix $W$ with blocks $W_1, \ldots, W_m$, $m
\geqslant 3$. This means that the $n$ nodes are partitioned into $m$ sets
$I_1, \ldots, I_m$, and the submatrix of $X$ that has rows and columns in
$I_j$ is constrained to equal $W_j$. So $X$ looks like
\[ \text{\begin{tabular}{l|llll|}
     \multicolumn{1}{c}{} & \multicolumn{1}{c}{$I_1$} & 
     \multicolumn{1}{c}{$I_2$} & \multicolumn{1}{c}{$\ldots$} & 
     \multicolumn{1}{c}{$I_m$} \\ \cline{2-5}
     $I_1$ & $W_1$ &  &  & \\
     $I_2$ &  & $W_2$ &  & \\
     $\vdots$ &  &  & $\ddots$ & \\
     $I_m$ &  &  &  & $W_m$ \\ \cline{2-5}
   \end{tabular}} \]
where the rest of the entries are determined by the $u$-constraints and, as
previously, are given by $s \lambda_i \lambda_j$. Thus for the nodes in the
set $I_1$ we have the equations
\begin{eqnarray*}
  s \lambda_1 ( \text{sum of } \lambda_j, j \notin I_1) & = & u_1 - (
  \text{sum of first row of } W_1),\\
  s \lambda_2 ( \text{sum of } \lambda_j, j \notin I_1) & = & u_2 - (
  \text{sum of second row of } W_1),
\end{eqnarray*}
etc. Let $\lambda_{I_1}$ denote $\sum_{i \in I_1} \lambda_i$, and similarly
for $\lambda_{I_2}$, etc. Also let $\lambda = \lambda_{I_1} + \cdots +
\lambda_{I_m}$. Then the above equations can be written as
\[ s \lambda_1 (\lambda - \lambda_{I_1}) = u_1 - w_{1 I_1}, \hspace{1em} s
   \lambda_2 (\lambda - \lambda_{I_2}) = u_2 - w_{2 I_1}, \hspace{1em} \ldots
\]
where the meaning of the additional notation should be clear. If we now add
these equations by sides, the result can be written compactly as
\[ \lambda_{I_1} (\lambda - \lambda_{I_1}) = r_{I_1}, \hspace{2em}
   \tmop{where} \hspace{1em} r_{I_1} = (u_{I_1} - w_{I_1 I_1}) / s, \]
and where subscripts that are sets indicate summation over the respective
sets. If we do the same thing for the rows in $I_2, \ldots, I_m$, we arrive at
the system of equations
\[ \lambda_{I_1} (\lambda - \lambda_{I_1}) = r_{I_1}, \hspace{1em}
   \lambda_{I_2} (\lambda - \lambda_{I_2}) = r_{I_2}, \hspace{1em} \ldots,
   \hspace{1em} \lambda_{I_m} (\lambda - \lambda_{I_m}) = r_{I_m}, \]
which has exactly the form (\ref{eq:lsys2}) except that here the $r_{I_i}$
don't sum to 1, but to
\[ \sigma \; = \; 1 - \frac{1}{s}  \sum_{i = 1}^m w_{I_i I_i}  \; < \; 1. \]
Of course, the $u_{I_i}$ and $w_{I_i I_i}$ are assumed to satisfy the
consistency condition (\ref{eq:cons}). Proceeding just as in
{\S}\ref{sec:fixd}, we have
\[ \lambda_{I_j} = \frac{1}{2}  \Bigl( \lambda - \sqrt{\lambda^2 - 4 r_{I_j}} 
   \Bigr) \]
so that $\lambda$ is the root of the equation
\begin{equation}
  \label{eq:lambda2} \sqrt{1 - 4 r_{I_1} / \lambda^2} + \cdots + \sqrt{1 - 4
  r_{I_m} / \lambda^2} \; = \; m - 2,
\end{equation}
about which we have a generalization of Proposition \ref{prop:eqlambda}:

\begin{proposition}
  \label{prop:eqlambda2}Suppose that $r_{I_1} + \cdots + r_{I_m} = \sigma <
  1$, and each $r_{I_j}$ is in $(0, \sigma / 3)$. Then for $m \geqslant 3$
  equation (\ref{eq:lambda2}) has a root in $(2 \sqrt{r_{\max}}, 4
  \sqrt{\sigma} / 3)$, where $r_{\max}$ is the largest of the $r_{I_j}$.
\end{proposition}

Given the root $\lambda$ of (\ref{eq:lambda2}), if $i \in I_k$, $\lambda_i$ is
given by $2 r_i / (\lambda + \sqrt{\lambda^2 - 4 r_{I_k}})$. But this
expression also equals $r_i (\lambda - \sqrt{\lambda^2 - 4 r_{I_k}}) / (2
r_{I_k})$. So the solution to our problem is: for $i \in I_k, j \in I_{\ell}$,
$k \neq \ell$,
\begin{equation}
  \begin{array}{c}
    \displaystyle
    \hat{x}_{i j} \; = \frac{s}{4}  \frac{\lambda^2 r_i r_j}{r_{I_k}
    r_{I_{\ell}}}  \bigl( 1 - \sqrt{1 - 4 r_{I_k} / \lambda^2}  \bigr) \bigl(
    1 - \sqrt{1 - 4 r_{I_{\ell}} / \lambda^2}  \bigr),\\[2ex]
    \displaystyle
    r_i = \frac{u_i - w_{i I_k}}{s}, \hspace{1em} r_{I_k} = \sum_{i \in I_k}
    r_i \; = \; \frac{u_{I_k} - w_{I_k I_k}}{s} .
  \end{array} \label{eq:bdiag}
\end{equation}
Suppose that all blocks are of size 1, so $m = n$ and the constraints are
$x_{i i} = w_{i i}$. Then it is easily seen that (\ref{eq:bdiag}) gives the
same result as (\ref{eq:zdiag}). An analytical approximation to the solution
of (\ref{eq:lambda2}), and to the matrix elements themselves, can be found by
the power series (\ref{eq:ps}).

Finally, the solution (\ref{eq:bdiag}) holds even when the $u_i$ are upper
bounds on the row and column sums. In that case Corollary \ref{cor:H} requires
$\lambda_{I_j} \leqslant 1$, which holds if $\forall j, 2 \sqrt{r_{I_j} }
\leqslant \lambda \leqslant 2$. But this last condition obtains by virtue of
Proposition \ref{prop:eqlambda2}.

\section{Conclusion}

\label{sec:concl}Table \ref{tab:summary} summarizes the problems for which we
obtained results in this paper. We saw that the most likely/{\maxent} matrices
exhibit as much independence, symmetry, and uniformity as possible subject to
the available information or constraints. Further, they are robust with
respect to changes in the information/constraints. Lastly, given independent
constraints on the rows and columns, the matrix elements have a ``product of
independent factors'' form, unless some of them are fixed, in which case the
independence disappears.

\begin{table}[h]
  \centering
  \begin{tabular}{|l|} \hline
    {\tmem{Rectangular matrices/contingency tables}}\\ \hline
    Given row sums and some column sums\\
    Bounds on row sums\\
    Total sum and bounds on row sums\\
    Bounds on total sum and row sums\\
    Bounds on row and column sums\\
    Bounds on row sums and on individual elements\\ \hline
    {\tmem{Square matrices with symmetric information}}\\ \hline
    Total sum and bounds on row and column sums\\
    Given row sums and partially-fixed diagonal,\\
    {\hspace{1em}}with extension to 3d matrices\\
    Given row sums and fixed diagonal blocks \\ \hline
  \end{tabular}
  \caption{\label{tab:summary}Summary of cases solved.}
\end{table}

The types of constraints that we considered were relatively simple, as befits
an initial exploration of the space of analytical solutions. The aim was to
have enough basic results to establish a framework for further investigations,
perhaps motivated by constraints arising in concrete problems.

Finally, even though we used the discrete balls--and--boxes framework
throughout, all that is said in this paper applies also to deriving
2-dimensional {\tmem{discrete probability distributions}} from incomplete
information, if we think of the balls as ``probability quanta'' thrown into
the boxes. Jaynes {\cite{JL}} calls this the ``Wallis derivation'' of
{\maxent} probability distributions.

\paragraph{Acknowledgments}Thanks to my colleagues Howard Karloff and N.J.A.
Sloane for interesting and helpful discussions.

\appendix\section{Auxiliary results and Proofs}

\subsection{Optimal solutions of concave programs}

\label{sec:concave}We review some standard terminology and results.

Suppose $C$ is a convex set in $\mathbb{R}^n$. A {\tmem{concave program}} is
the problem of maximizing a concave function $f$ on this set, subject to a
number of equality and inequality constraints:
\begin{equation}
  \label{eq:cp} \begin{array}{c}
    \max_{x \in C} f (x) \hspace{2em} \tmop{subject} \tmop{to}\\
    g_i (x) = 0, \hspace{1em} i = 1, \ldots, \ell, \hspace{2em} h_j (x)
    \leqslant 0, \hspace{1em} j = 1, \ldots, m,
  \end{array}
\end{equation}
where the $g_i$ are \tmtextit{linear} on $C$ (and assumed linearly
independent) and the $h_j$ are \tmtextit{convex} on $C$. All $x \in C$
satisfying the constraints are called \tmtextit{feasible}. The
{\tmem{Lagrangean function}} associated with the concave program (\ref{eq:cp})
is
\begin{equation}
  \label{eq:lagr} \Phi (x, \alpha, \beta) = f (x) - \sum_i \alpha_i g_i (x) -
  \sum_j \beta_j h_j (x) .
\end{equation}
The following result (Theorem 2.30 in {\cite{ADSZ}}, or {\S}5.5.3 of
{\cite{BV}}) gives sufficient conditions for solving a concave program in
which all functions are differentiable on (the interior of) $C$:

\begin{theorem}
  \label{th:cp}If $x^{\ast}$ is feasible, and there are $\alpha^{\ast},
  \beta^{\ast}$ such that
  \[ \nabla_x \Phi (x^{\ast}, \alpha^{\ast}, \beta^{\ast}) = 0, \hspace{2em}
     \beta^{\ast}_j h_j (x^{\ast}) = 0 \hspace{1em} \tmop{and} \hspace{1em}
     \beta^{\ast}_j \geqslant 0 \hspace{1em} \forall j, \]
  then $x^{\ast}$ solves the concave program (\ref{eq:cp}).
\end{theorem}

Also recall that if a strictly concave function on a convex set has a maximum,
the maximizing point is unique (Theorem 2.22 in {\cite{ADSZ}}).

\begin{corollary}
  \label{cor:H}Suppose the function $f$ in (\ref{eq:lagr}) is the entropy, and
  all the constraints are linear and involve coefficients that are either 0 or
  1. Then the elements of $x^{\ast}$ have the form
  \[ x^{\ast}_k \; = \; \prod_{i \in E_k} \alpha'_i  \prod_{j \in I_k}
     \beta'_j, \hspace{2em} \text{where} \hspace{1em} \alpha'_i > 0, \beta'_j
     \in (0, 1], \]
  where $E_k$ is the set of indices of the equalities $g_i$ in which $x_k$
  appears, and $I_k$ is the set of indices of the inequalities $h_j$ where
  $x_k$ appears. The $j$-th inequality constraint can be satisfied either as a
  strict inequality or as an equality, and we must have
  \begin{equation}
    \label{eq:multprime} h_j (x^{\ast}) \ln \beta'_j = 0.
  \end{equation}
\end{corollary}

\begin{corollary}
  \label{cor:G}If the function $f$ in (\ref{eq:lagr}) is the entropy
  difference function $G$ of (\ref{eq:G}) and the constraints are as in
  Corollary \ref{cor:H}, then
  \[ x^{\ast}_k \; = \Bigl( \sum_{1 \leqslant \ell \leqslant n}
     x^{\ast}_{\ell} \Bigr) \prod_{i \in E_k} \alpha'_i  \prod_{j \in I_k}
     \beta'_j, \hspace{2em} \alpha'_i > 0, \beta'_j \in (0, 1], \]
  where the $\beta_j'$ must satisfy (\ref{eq:multprime}).
\end{corollary}

\subsection{Proofs for {\S}\ref{sec:bounds1}}

\subsection*{Proof of Lemma \ref{le:ub}}

The Lagrangean is
\[ \Phi = - \sum_i x_i \ln x_i - \lambda_1 (x_1 - b_1) - \cdots - \lambda_n
   (x_n - b_n) - \mu (x_1 + \cdots + x_n - a) . \]
Setting $\nabla \Phi$ to 0, we have for all $i$
\begin{equation}
  \label{eq:ub1} x_i \; = \; e^{- \lambda_i - \mu - 1} \; \rightsquigarrow \;
  \lambda_i \mu .
\end{equation}
By Corollary \ref{cor:H}, for a point $(x_1, \ldots, x_n)$ given by
(\ref{eq:ub1}) to solve the problem the following must hold
\begin{enumeratealpha}
  \item $(x_1, \ldots, x_n) \tmop{must} \tmop{be} \tmop{feasible}$,
  
  \item By (\ref{eq:multprime}), we must have $\lambda_i \in (0, 1]$ for all
  $i$, and $(x_i - b_i) \ln \lambda_i = 0$.
\end{enumeratealpha}
Now arrange the $b_i$ and $x_i$ as stated in part (i) of the lemma. Consider
the solution
\begin{equation}
  \begin{array}{ccc}
    x_i = b_i = \lambda_i \mu, & \text{with $\lambda_i \leqslant 1$}, & i = 1,
    \ldots, k\\
    x_i = \mu, & \text{with $\lambda_i = 1$,} & i = k + 1, \ldots, n
  \end{array} \label{eq:ub2}
\end{equation}
in accordance with (b) above, where $k$ is as yet undetermined. Putting
(\ref{eq:ub2}) into the equality constraint we get $b_1 + \cdots + b_k + (n -
k) \mu = a$. It follows that
\begin{equation}
  \label{eq:ub3} x_{k + 1} = \cdots = x_n = \mu = \frac{a - (b_1 + \cdots +
  b_k)}{n - k} .
\end{equation}
Now let $k$ be chosen as in part (ii) of the lemma. Then the solution $(x_1,
\ldots, x_n)$ given by (\ref{eq:ub2}), (\ref{eq:ub3}) is feasible as required
in (a) above: by the definition of $k$, $b_1 + \cdots + b_{k + 1} + (n - k -
1) b_{k + 1} > a$, which is equivalent to $\mu < b_{k + 1}$.

To satisfy (b), we need to check that $\lambda_i \leqslant 1$ for $i = 1,
\ldots, k$. From (\ref{eq:ub2}) and (\ref{eq:ub3}),
\[ \lambda_i = \frac{(n - k) b_i}{a - (b_1 + \cdots + b_k)} \hspace{2em}
   \tmop{and} \hspace{2em} \lambda_i \leqslant 1 \hspace{1em} \Leftrightarrow
   \hspace{1em} a - (b_1 + \cdots + b_k) \geqslant (n - k) b_i . \]
But this last condition holds $\forall i \leqslant k$ by the definition of
$k$. We have found a solution $x^{\ast}$, and because the entropy function is
strictly concave, this solution is unique and we are done. It remains to show
that it is possible to find a $k$ as required in part (ii) of the lemma. This
is done in Proposition \ref{prop:k} below.

\begin{proposition}
  \label{prop:k}Given $b_0 = 0 < b_1 \leqslant b_2 \leqslant \cdots \leqslant
  b_n$ and $0 < a \leqslant b_1 + \cdots + b_n$, there is a $k \in \{0,
  \ldots, n\}$ s.t. the inequality
  \[ a - (b_1 + \cdots + b_j) \geqslant (n - j) b_j \]
  holds for all $j \leqslant k$ and for no larger $j$.
\end{proposition}

\begin{proof}
  Consider the function $\varphi (j) = a - (b_1 + \cdots + b_j) - (n - j)
  b_j$, $j \in \{0, 1, \ldots, n\}$. It is easy to see that $\varphi (j)
  \geqslant \varphi (j + 1)$ for all $j$, so this function is monotone
  decreasing. Further, $\varphi (0) = a > 0$ and $\varphi (n) = a - (b_1 +
  \cdots + b_n) \leqslant 0$. So there is a $k \leqslant n$ s.t. $\varphi (j)
  \geqslant 0$ for $j \leqslant k$, and $\varphi (j) < 0$ for $j > k$, as
  claimed. Note that $k = n$ iff $b_1 + \cdots + b_n = a$.
\end{proof}

\subsection{Proofs for {\S}\ref{sec:boundsrc}}

\begin{proposition}
  \label{prop:G}The function
  \begin{eqnarray*}
    G (x_1, \ldots, x_n) & = & \Bigl( \sum_i x_i \Bigr) \ln \Bigl( \sum_i x_i
    \Bigr) - \sum_i x_i - \sum_i (x_i \ln x_i - x_i)\\
    & = & \Bigl( \sum_i x_i \Bigr) \ln \Bigl( \sum_i x_i \Bigr) - \sum_i x_i
    \ln x_i
  \end{eqnarray*}
  is concave over the domain $x_1 > 0, \ldots, x_n > 0$.
\end{proposition}

This is probably known somewhere in the information theory literature, but I
don't know where. So a proof is presented below.

\begin{proof}
  By Theorem 2.14 of {\cite{ADSZ}} it suffices to show that $\mathcal{H}(x) =
  \nabla^2 G (x)$, the Hessian of $G$, is negative semi-definite. We find
  \[ \mathcal{H}(x) \; = \; \frac{1}{x_1 + \cdots + x_n} U_n - \tmop{diag}
     \left( \frac{1}{x_1}, \ldots, \frac{1}{x_n} \right), \]
  where $U_n$ is a matrix all of whose entries are 1, and for an arbitrary
  vector $y = (y_1, \ldots, y_n)$ we must have $y^T \mathcal{H y} \leqslant
  0$. To establish this, first write $\mathcal{H}$ as
  \[ \; \mathcal{H}(x) \; = \; \frac{1}{x_1 + \cdots + x_n}  \left( U_n -
     \tmop{diag} \left( \frac{x_1 + \cdots + x_n}{x_1}, \ldots, \frac{x_1 +
     \cdots + x_n}{x_n} \right) \right) . \]
  Now define $\xi_i = x_i / (x_1 + \cdots + x_n)$. The condition $y^T
  \mathcal{H y} \leqslant 0$ is then equivalent to
  \begin{equation}
    \label{eq:yconv} (y_1 + \cdots + y_n)^2 \; \leqslant \; y_1^2 / \xi_1 +
    \cdots + y_n^2 / \xi_n,
  \end{equation}
  where the $\xi_i$ are positive and sum to 1. The truth of (\ref{eq:yconv})
  follows from the fact that $y^2_1 / \xi_1 + \cdots + y^2_n / \xi_n$ is a
  convex function of $\xi_1, \ldots, \xi_n$ over the domain $\xi_1 > 0,
  \ldots, \xi_n > 0$, and its minimum under the constraint $\xi_1 + \cdots +
  \xi_n = 1$ occurs at $\xi^{\ast}_i = y_i / (y_1 + \cdots + y_n)$. So the
  least value of the r.h.s. of (\ref{eq:yconv}) as a function of $\xi_1,
  \ldots, \xi_n$ is $(y_1 + \cdots + y_n)^2$.
\end{proof}

\subsection*{Proof of Proposition \ref{prop:rcb}}

We first give a straightforward proof assuming that $x_{i j} \in \mathbb{N}$.
Suppose there is a matrix $X$ s.t. for some $i, j$ row $i$ sums to less than
$u_i$ and column $j$ to less than $v_j$. Further, let $X$ have total sum $s$.
Consider the matrix $X'$ formed by adding 1 to $x_{i j}$. First, if $X$
satisfies the constraints, so does $X'$. Second, $\#(X' |I) /\#(X|I) = (s + 1)
/ (x_{i j} + 1) > 1$. Thus $X'$ has more realizations than $X$, and so $X$
cannot be the most likely matrix $\hat{X}$.

To give a proof assuming that the elements of $X$ are non-negative reals, by
Corollary \ref{cor:G} we must have $\hat{x}_{i j} = ( \sum_{k, l} \hat{x}_{k
l}) \lambda_i \mu_j$. Now if there is a pair $i, j$ s.t. $\sum_j \hat{x}_{i j}
- u_i < 0$ and $\sum_i \hat{x}_{i j} - v_j < 0$, we must have $\lambda_i =
\mu_j = 1$, so $\hat{x}_{i j} = \sum_{k, l} \hat{x}_{k l}$. Thus all other
elements of $\hat{X}$must be 0. Further, $\hat{x}_{i j} \leqslant \min (u_i,
v_j)$. But it is easy to see that this matrix cannot have the most
realizations.

\subsection*{Proof of Proposition \ref{prop:rcb2}}

Consider the function $\varphi (\ell) = u - (v_1 + \cdots + v_{\ell}) - (n -
\ell) v_{\ell + 1}$. It is easy to check that $\varphi (\ell) \nearrow$ as
$\ell \nearrow$. Further, $\varphi (0) = u - n v_1 \geqslant 0$ if $u / n
\geqslant v_1$. Finally, $\varphi (n - 1) = u - (v_1 + \cdots + v_n) < 0$.
Thus there is a least $\ell$, s.t. $\varphi (\ell) < 0$, $1 \leqslant \ell < n
- 1$, and $\varphi (\ell - 1) \geqslant 0$. Let that $\ell$ be $k$. The two
conditions $\varphi (k) < 0$ and $\varphi (k - 1) \geqslant 0$ establish what
is claimed.

\subsection{Proofs for {\S}\ref{sec:bounds3}}

The following result is a variation of Lemma \ref{le:ub}: it says that the
most likely vector with sum bounded by $a$ and elements bounded by the vector
$b$ is the {\maxent} vector with sum equal to $a$ and elements bounded by $b$.

\begin{lemma}
  \label{le:ub2}The most likely vector $x^{\ast} = (x_1^{\ast}, \ldots,
  x^{\ast}_m)$ satisfying $\forall i \; 0 \leqslant x_i \leqslant b_i$ and
  $x_1 + \cdots + x_m \leqslant a$, $a, b_i \in \mathbb{N}$, is found as
  follows. If $a \leqslant b_1 + \cdots + b_m$, the inequality in this
  constraint can be replaced by equality and then $x^{\ast}$ is given by Lemma
  \ref{le:ub}. If $a > b_1 + \cdots + b_m$, then $x^{\ast} = (b_1, \ldots,
  b_m)$.
\end{lemma}

\begin{proof}
  First we reduce the problem in $\mathbb{N}$ to another problem in
  $\mathbb{N}$. Suppose that $a \leqslant b_1 + \cdots + b_m$. Let $y = (y_1,
  \ldots, y_m), y_i \in \mathbb{N}$ be the most likely vector
  \tmtextit{summing} to $\alpha \leqslant a - 1$ and satisfying $y_i \leqslant
  b_i$. Pick a $y_j$ s.t. $y_j < b_j$; this exists because $y$ sums to
  $\alpha$, which is strictly less than $b_1 + \cdots + b_m$. But then the
  vector $y' = (y_1, \ldots, y_{j - 1}, y_j + 1, y_{j + 1}, \ldots, y_m)$ sums
  to $\alpha + 1$, satisfies the $b$-constraints, and by the argument given in
  {\S}\ref{sec:bounds2}, $\#(y' \mid \alpha + 1) >\#(y \mid \alpha)$. So by
  increasing the allowed sum $\alpha$ we get a more likely vector. It follows
  that the most likely vector $x^{\ast}$ in $\mathbb{N}$ satisfying the
  constraints sums to exactly $a$, and (an approximation in $\mathbb{R}$) can
  therefore be found by Lemma \ref{le:ub}.
  
  Now let $a > b_1 + \cdots + b_m$. In that case the $a$-constraint is
  irrelevant and we have precisely the problem solved in {\S}\ref{sec:bounds2}
  for a matrix; so $x^{\ast} = (b_1, \ldots, b_m)$.
\end{proof}

\subsection{\label{sec:appsymm}Proofs for {\S}\ref{sec:fixd}}

\subsubsection{Proof of Proposition \ref{prop:eqlambda}}

We already noted that the function
\[ f (\lambda) = \sqrt{1 - 4 r_1 / \lambda^2} + \cdots + \sqrt{1 - 4 r_m /
   \lambda^2} - 2 (r_{m + 1} + \cdots + r_n) / \lambda^2 - (m - 2) \]
is monotone increasing for any $m \leqslant n$. We will now show that $f (4 /
3) > 0$ and $f (2 \sqrt{r_{\max}}) \leqslant 0$.

\paragraph{$f > 0$ at 4/3}

This reduces to showing that
\begin{equation}
  \label{eq:fl} \sqrt{1 - 9 / 4 r_1} + \cdots + \sqrt{1 - 9 / 4 r_m} - 9 / 8
  (r_{m + 1} + \cdots + r_n) \; > \; m - 2.
\end{equation}
The l.h.s. has the form $\sum_i \varphi_i (r_i)$ where $\varphi_i (\cdot)$ is
concave, so it is a concave function of $r_1, \ldots, r_n$ (Prop. 2.16 of
{\cite{ADSZ}}) over the convex domain defined by $r_1 + \cdots + r_n = 1$ and
$0 < r_i \leqslant 1 / 3$. Therefore its minimum occurs on the boundary of the
domain ({\cite{ADSZ}}, Prop. 2.25.) The boundary consists of all points s.t.
three of the $r_i$ are 1/3 and the rest are 0. There are several cases to
consider. First, it is easy to check that (\ref{eq:fl}) holds for $m = 0$ and
$m = 1$.

Next let $m = 2$. What we want to prove reduces to $\sqrt{1 - 9/4 r_1} +
\sqrt{1 - 9/4 r_2} - 9/8 (r_3 + \cdots + r_n) > 0$. The possibilities for
the boundary are $r_1 = r_2 = r_3 = 1/3$, or $r_1 = 1 / 3, r_3 = r_4 = 1/3$, or
$r_3 = r_4 = r_5 = 1/3$, and the desired inequality holds under any 
of these conditions.

Lastly suppose that $m \geqslant 3$, and, without loss of generality, that
$r_1 = r_2 = r_3 = 1 / 3$. Then (\ref{eq:fl}) becomes $3 / 2 + m - 3 > m - 2$,
which is true. Next, let $r_1 = r_2 = 1 / 3$, $r_{m + 1} = 1 / 3$;
(\ref{eq:fl}) becomes $1 + m - 2 - 3 / 8 > m - 2$, which is also true. The
remaining two cases are $r_1 = 1 / 3$, $r_{m + 1} = r_{m + 2} = 1 / 3$, and
$r_{m + 1} = r_{m + 2} = r_{m + 3} = 1 / 3$, and (\ref{eq:fl}) holds for both.

\paragraph{$f \leqslant 0$ at $2 \sqrt{r_{\max}}$}

Without loss of generality we may assume that $r_{\max} = r_1$ because this
makes the notation simpler. Then $f(2 \sqrt{r_{\max}}) < 0$ reduces to
establishing
\begin{equation}
  \label{eq:o/o} \sqrt{1 - r_2 / r_1} + \cdots + \sqrt{1 - r_m / r_1} -
  \frac{r_{m + 1} + \cdots + r_n}{2 r_1} \leqslant m - 2.
\end{equation}
We will find the maximum of the function on the l.h.s., treating $r_1$ as
known for the moment. Using Theorem \ref{th:cp}, the l.h.s. is a concave
function of $r_2, \ldots, r_n$, and under the constraint $r_2 + \cdots + r_n =
1 - r_1$ it has a unique maximum at the point determined by
\[ \frac{1}{2 r_1} = \frac{1}{2 r_1} \Bigl( 1 - \frac{r_2}{r_1} \Bigr)^{- 3 /
   2} = \; \cdots \; = \frac{1}{2 r_1} \Bigl( 1 - \frac{r_m}{r_1} \Bigr)^{- 3
   / 2} . \]
Thus the maximum occurs at the point $r_2 = \cdots = r_m = 0$ and $r_{m + 1} +
\cdots + r_n = 1 - r_1$, where the value of the function is $m - 1 - (1 - r_1)
/ (2 r_1)$. Therefore (\ref{eq:o/o}) will hold iff $r_1 \leqslant 1 / 3$.

The above proof assumed that $m < n$. When $m = n$, (\ref{eq:o/o}) becomes
\begin{equation}
  \label{eq:o/o/o} \sqrt{1 - r_2 / r_1} + \cdots + \sqrt{1 - r_n / r_1}
  \leqslant n - 2.
\end{equation}
As before, the l.h.s. is a concave function for fixed $r_1$, and its maximum
occurs at $r_2 = \cdots = r_n = (1 - r_1) / (n - 1)$. Thus (\ref{eq:o/o/o}) holds
if
\[ (n - 1) \sqrt{1 - \frac{1 - r_1}{(n - 1) r_1}} \leqslant n - 2, \]
which is true if $r_1 \leqslant 1 / 3$.

\paragraph{Proof of Proposition \ref{prop:eqlambda2}}

We re-use the proof of Proposition \ref{prop:eqlambda}. Define $f (\lambda) =
\sqrt{1 - 4 r_{I_1} / \lambda^2} + \cdots + \sqrt{1 - 4 r_{I_m} / \lambda^2} -
(m - 2)$. Setting $\rho_i = r_{I_i} / \sigma$, this becomes
\[ f (\lambda) = \sqrt{1 - 4 \sigma \rho_1 / \lambda^2} + \cdots + \sqrt{1 - 4
   \sigma \rho_m / \lambda^2} - (m - 2), \hspace{2em} \sum_i \rho_i = 1. \]
Then $f (4 \sqrt{\sigma} / 3) > 0$ is equivalent to $\sqrt{1 - 9 / 4 \rho_1} +
\cdots + \sqrt{1 - 9 / 4 \rho_m} > m - 2$; but this is a special case of
(\ref{eq:fl}). It remains to show that $f (2 \sqrt{r_{\max}}) = f (2
\sqrt{\sigma \rho_{\max}}) \leqslant 0$. Assuming w.l.o.g. that $\rho_{\max} =
\rho_1$, this reduces to $\sqrt{1 - \rho_2 / \rho_1} + \cdots + \sqrt{1 -
\rho_2 / \rho_m} \leqslant m - 2$, which follows from (\ref{eq:o/o/o}).


\begin{thebibliography}{ACR+06}
  \bibitem[ACR+06]{Alderson2006}D. Alderson, H. Chang, M. Roughan, S. Uhlig,
  and W. Willinger. {\newblock}The Many Facets of Internet Topology and
  Traffic. {\newblock}\tmtextit{Networks and Heterogeneous Media}, 1(4), 2006.
  
  \bibitem[ADSZ88]{ADSZ}M. Avriel, W.E. Diewert, S. Schaible, and I. Zang.
  {\newblock}\tmtextit{Generalized Concavity}. {\newblock}Plenum Press, 1988.
  
  \bibitem[BD08]{Bilich2008}F. Bilich and R.~ DaSilva. {\newblock}Maximum
  Entropy Principle for Transportation. {\newblock}In \tmtextit{Bayesian
  Inference and Maximum Entropy Methods in Science and Engineering, 28}.
  American Institute of Physics (AIP), 2008.
  
  \bibitem[BP94]{mat}A. Berman and R.J. Plemmons.
  {\newblock}\tmtextit{Nonnegative Matrices in the Mathematical Sciences}.
  {\newblock}SIAM Press, 1994.
  
  \bibitem[BV04]{BV}S. Boyd and L. Vandenberghe. {\newblock}\tmtextit{Convex
  Optimization}. {\newblock}Cambridge, 2004.
  
  \bibitem[CG02]{Cmiel2002}A. Cmiel and H. Gurgul. {\newblock}Application of
  maximum entropy principle in key sector analysis.
  {\newblock}\tmtextit{Systems Analysis Modelling Simulation},
  42(9):1361--1376, 2002.
  
  \bibitem[CG06]{Caticha2006}A. Caticha and A. Giffin. {\newblock}Updating
  Probabilities. {\newblock}In \tmtextit{26th International Workshop on
  Bayesian Inference and Maximum Entropy Methods in Science and Engineering},
  2006.
  
  \bibitem[ES90]{grav}S. Erlander and N.F. Stewart. {\newblock}\tmtextit{The
  Gravity Model in Transportation Analysis: Theory and Extensions}.
  {\newblock}VSP, Utrecht, The Netherlands, 1990.
  
  \bibitem[Goo63]{GoodME}I.J. Good. {\newblock}Maximum Entropy for Hypothesis
  Formulation, Especially for Multidimensional Contingency Tables.
  {\newblock}\tmtextit{Annals of Mathematical Statistics}, 34(3):911--934,
  1963.
  
  \bibitem[Hen88]{Hen}P. Henrici. {\newblock}\tmtextit{Applied and
  Computational Complex Analysis, Vol. 1}. {\newblock}John Wiley, 1988.
  
  \bibitem[Jay03]{JL}E.T. Jaynes. {\newblock}\tmtextit{Probability Theory: The
  Logic of Science}. {\newblock}Cambridge University Press, 2003.
  
  \bibitem[KK92]{KK}J.N. Kapur and H.K. Kesavan. {\newblock}\tmtextit{Entropy
  Optimization Principles with Applications}. {\newblock}Academic Press, 1992.
  
  \bibitem[KMO93]{Ku-Mahamud1993}K. Ku-Mahamud and A. Othman. {\newblock}Model
  reduction of general queueing networks. {\newblock}\tmtextit{International
  Journal of Systems Science}, 24(1):183--192, 1993.
  
  \bibitem[KO08]{KO}S.K. Korotky and K.N. Oikonomou. {\newblock}Scaling of
  Most-Likely Traffic Patterns of Hose- and Cost-Constrained Ring and Mesh
  Networks. {\newblock}\tmtextit{Journal of Optical Networking}, 7:550--563,
  June 2008.
  
  \bibitem[KT92]{Kouvatsos1992}D. Kouvatsos and P. Tomaras.
  {\newblock}Multilevel aggregation of central server models: a minimum
  relative entropy approach. {\newblock}\tmtextit{International Journal of
  Systems Science}, 23(5):713--739, 1992.
  
  \bibitem[MAX98]{maxent-kluwer}\tmtextit{Maximum Entropy and Bayesian Methods
  in Science and Engineering}. Kluwer Academic Publishers, 1985-1998.
  
  \bibitem[MAX09]{maxent-aip}\tmtextit{Bayesian Inference and Maximum Entropy
  Methods in Science and Engineering}. American Institute of Physics (AIP),
  1999-2009.
  
  \bibitem[Ros83]{JaynesCP}R.D. Rosenkrantz, editor. {\newblock}\tmtextit{E.
  T. Jaynes: Papers on Probability, Statistics, and Statistical Physics.}
  {\newblock}D. Reidel [Kluwer], Dordrecht, The Netherlands, 1983.
  
  \bibitem[Sen91]{Sengupta1991}J.K. Sengupta. {\newblock}Maximum entropy in
  applied econometric research. {\newblock}\tmtextit{International Journal of
  Systems Science}, 22(10):1941--1951, 1991.
  
  \bibitem[Ski89]{Skilling1989}J. Skilling. {\newblock}Classic maximum
  entropy. {\newblock}In J. Skilling, editor, \tmtextit{Maximum Entropy and
  Bayesian Methods}. Kluwer Academic, 1989.
  
  \bibitem[Som67]{Som}A. Sommerfeld. {\newblock}\tmtextit{Thermodynamics and
  Statistical Mechanics}. {\newblock}Lectures on Theoretical Physics, Vol. 5.
  Academic Press, 1967.
  
  \bibitem[TJI02]{Trivedi2002}S. Trivedi, B. Jones, and S. Iyengar.
  {\newblock}Why k-systems methodology works. {\newblock}\tmtextit{Systems
  Analysis Modelling Simulation}, 42(1):23--31, 2002.
  
  \bibitem[Tri69]{Tr69}M. Tribus. {\newblock}\tmtextit{Rational Descriptions,
  Decisions and Designs}. {\newblock}Pergamon Press, 1969.
  
  \bibitem[ZRLD05]{ZRLD}Y. Zhang, M. Roughan, C. Lund, and D. Donoho.
  {\newblock}Estimating Point-to-Point and Point-to-Multipoint Traffic
  Matrices: An Information-Theoretic Approach. {\newblock}\tmtextit{IEEE/ACM
  Transactions on Networking}, 13(5), 2005.
\end{thebibliography}
\end{document}